\documentclass[11pt]{article}
\usepackage{graphics,graphicx,tikz,color,float, hyperref}
\usepackage{mathtools}
\usepackage{amsmath}
\usepackage{amssymb}
\usepackage{amsfonts}
\usepackage{amsthm}
\usepackage{fullpage}
\usepackage {comment}

\usepackage{xcolor}
\newcommand{\mycomment}[1]{\textup{{\color{red}#1}}}

\newtheorem{fact}{Fact}

\newcommand*{\nmasw}{\mathsf{ wqnm\mhyphen state}}
\newcommand*{\qpasw}{\mathsf{ wqpa\mhyphen state}}

\newcommand*{\qmas}{\mathsf{qma\mhyphen state}}
\newcommand{\tab}{\hspace*{2em}}
\newcommand{\beq}{\begin{equation}}
	\newcommand{\enq}{\end{equation}}
\newcommand{\bel}{\begin{lemma}}
	\newcommand{\enl}{\end{lemma}}
\newcommand{\bet}{\begin{theorem}}
	\newcommand{\ent}{\end{theorem}}

\newcommand{\tr}{\mathrm{Tr}}

\newcommand{\ketbra}[1]{|#1\rangle\langle#1|}

\newcommand{\err}{\mathrm{err}}
\newcommand{\eps}{\varepsilon}

\newcommand*{\cH}{\mathcal{H}}

\newcommand*{\cX}{\mathcal{X}}

\newcommand*{\cZ}{\mathcal{Z}}
\newcommand*{\cE}{\mathcal{E}}

\newcommand*{\IP}{\mathsf{IP}}
\newcommand{\Ext}{\mathsf{Ext}}
\newcommand{\Q}{\mathcal{Q}}
\newcommand{\supp}{\mathrm{supp}}
\newcommand{\suppress}[1]{}

\newcommand{\defeq}{\ensuremath{ \stackrel{\mathrm{def}}{=} }}
\newcommand{\F}{\mathrm{F}}

\newcommand {\br} [1] {\ensuremath{ \left( #1 \right) }}

\newcommand {\minusspace} {\: \! \!}
\newcommand {\smallspace} {\: \!}
\newcommand {\fn} [2] {\ensuremath{ #1 \minusspace \br{ #2 } }}

\newcommand {\relentalpha} [3] {\fn{\mathrm{D}_{#3}}{#1 \middle\| #2}}
\newcommand {\dmax} [2] {\fn{\mathrm{D}_{\max}}{#1 \middle\| #2}}

\newcommand {\mutinf} [2] {\fn{\mathrm{I}}{#1 \smallspace : \smallspace #2}}
\newcommand {\imax}{\ensuremath{\mathrm{I}_{\max}}}

\newcommand {\condmutinf} [3] {\mutinf{#1}{#2 \smallspace \middle\vert \smallspace #3}}
\newcommand {\hminn} [2] {\fn{ \mathrm{ \tilde{H} }_{\min}}{#1 \middle | #2}}
\newcommand {\hmin} [2] {\fn{ \mathrm{H }_{\min}}{#1 \middle | #2}}
\newcommand {\hmineps} [3] {\fn{\mathrm{H}^{#3}_{\min}}{#1 \middle | #2}}

\newcommand {\id} {\ensuremath{\mathrm{I}}}

\newcommand {\Hmin}{\mathrm{H}_{\min}}

\newcommand{\QCC}{\mathrm{QCC}}

\newcommand*{\eff}{ \tilde{\mathsf{eff}}}

\newcommand*{\qpas}{\mathsf{qpa\mhyphen state}}
\newcommand*{\qcas}{\mathsf{qca\mhyphen state}}

\newcommand*{\qmras}{\mathsf{qMara\mhyphen state}}
\newcommand*{\qias}{\mathsf{qia\mhyphen state}}
\newcommand*{\geas}{\mathsf{gea\mhyphen state}}
\newcommand*{\qma}{\mathsf{qma}}
\newcommand*{\qpa}{\mathsf{qpa}}
\newcommand*{\qca}{\mathsf{qca}}
\newcommand*{\qbsa}{\mathsf{qbsa}}
\newcommand*{\qmra}{\mathsf{qMara}}
\newcommand*{\qia}{\mathsf{qia}}
\newcommand*{\gea}{\mathsf{gea}}
\newcommand*{\nma}{\mathsf{qnma}}
\newcommand*{\nmext}{\mathsf{nmExt}}

\newcommand{\mac}{\mathsf{MAC}}
\newcommand{\X}{\mathcal{X}}
\newcommand{\Y}{\mathcal{Y}}
\newcommand{\Z}{\mathcal{Z}}

\newcommand*{\cL}{\mathcal{L}}

\newcommand{\bra}[1]{\langle #1|}
\newcommand{\ket}[1]{|#1 \rangle}

\mathchardef\mhyphen="2D

\makeatletter
\newcommand*{\rom}[1]{\expandafter\@slowromancap\romannumeral #1@}
\makeatother

\mathchardef\mhyphen="2D

\newtheorem*{remark}{Remark}
\newtheorem{definition}{Definition}
\newtheorem{claim}{Claim}

\newtheorem{theorem}{Theorem}
\newtheorem{lemma}{Lemma}
\newtheorem{corollary}{Corollary}

\title{Quantum Measurement Adversary}
\author{
	Divesh Aggarwal~\footnote{Centre for Quantum Technologies and Department of Computer Science, National University of Singapore, \texttt{dcsdiva@nus.edu.sg}}\qquad
	Naresh Goud Boddu~\footnote{NTT Research, Sunnyvale, USA, \texttt{naresh.boddu@ntt-research.com}}\qquad
Rahul Jain~\footnote{Centre for Quantum Technologies and Department of Computer Science, 
  National University of Singapore and MajuLab, UMI 3654, Singapore,  \texttt{rahul@comp.nus.edu.sg}}\qquad
Maciej Obremski~\footnote{Center for Quantum Technologies, National University of Singapore, \texttt{obremski.math@gmail.com}}
}
\newfloat{Protocol}{tbhp}{lop}
\begin{document}
\maketitle

  \begin{abstract}
  	{\em Multi-source extractors} are functions that extract uniform randomness from multiple ({\em weak}) sources of randomness. Quantum multi-source extractors were considered by Kasher and Kempe ~\cite{KK10} (for the {\em quantum independent adversary} and the {\em quantum bounded storage adversary}), Chung, Li and Wu~\cite{CLW14} (for the {\em general entangled adversary}) and Arnon-Friedman, Portmann and Scholz~\cite{APS16} (for the {\em quantum Markov adversary}).

One of the main objectives of this work is to unify all the existing quantum multi-source adversary models. We propose two new models of adversaries: 1) the quantum measurement adversary ($\qma$), which generates side information using entanglement and on post-measurement and 2) the quantum communication adversary ($\qca$), which generates side information using entanglement and communication between multiple sources. We show that,
 \begin{enumerate}
     \item $\qma$ is the strongest adversary among all the known adversaries, in the sense that the side information of all other adversaries can be generated by $\qma$.
       \item The (generalized) inner-product function (in fact a general class of two-wise independent functions) continues to work as a good extractor with matching parameters as that of Chor and Goldreich~\cite{CG85}  against classical adversaries. 
    \item A {\em non-malleable extractor} proposed by Li~\cite{Li12c} (against classical adversaries) continues to be secure against quantum side information. This result implies a non-malleable extractor result of Aggarwal,  Chung,  Lin and  Vidick~\cite{ACLV18} with uniform seed. We strengthen their result via a completely different proof to make the non-malleable extractor of  Li secure against quantum side information even when the seed is not uniform.
     \item A modification (working with weak local randomness instead of uniform local randomness) of the Dodis and Wichs \cite{DW09} protocol for {\em privacy-amplification}  is secure against active quantum adversaries (those who arbitrarily modify the messages exchanged in the protocol). This strengthens on a recent result due to Aggarwal,  Chung,  Lin and  Vidick~\cite{ACLV18} which uses uniform local randomness.
   
    \item A tight {\em efficiency} lower bound for the (generalized) inner-product function (in fact a general class of two-wise independent functions).
 \end{enumerate}
  	\end{abstract}
  	\pagebreak
\section{Introduction}
Randomized algorithms are designed under the assumption that randomness used within the algorithms is uniformly distributed. However, random sources in nature are often not necessarily uniform (aka weak). Thus, it is important to understand how to extract uniform randomness from weak-sources. Extractors are functions that transform weak-sources to uniform randomness. Extractors have numerous applications including privacy-amplification, pseudo-randomness, derandomization,  expanders, combinatorics  and cryptography.
A general model of weak-sources is the so-called {\em min-entropy} source (please refer to Section~\ref{sec:prelims} for definitions of information-theoretic quantities). Let $n, t, m$ be positive integers and  $k, k_1, k_2, k_1', k_2', b_1, b_2, l$ be positive reals.

\suppress{
The min-entropy of a random variable $X$ is given by 
\[    \Hmin(X) =\min_{x \in     \text{supp}(X)       } \log\left(\frac{1}{\Pr(X=x)} \right).\]
}
\begin{definition}[\cite{ZUC90}]\label{weaksource}
	An $(n,k)$-source denotes an $n$-bit source $X\in \{ 0,1\}^n$ with the min-entropy bound $k$, i.e. $\Hmin(X)  \geq k$. 
\end{definition}
It can be argued that no deterministic function can extract even one uniform bit given an (arbitrary) $(n,k)$-source as long as $k \leq n-1$ \cite{CG85}. This lead to designing extractors using an additional uniform source (aka {\em seed}) called {\em seeded extractors}~\cite{FS04}.  Another approach is to  consider multiple independent weak-sources.  
\begin{definition}[\cite{ZUC90, CLW14}]\label{tsource}
	An $(n,k_1,k_2, \ldots, k_t)$-source denotes $t$ independent $n$-bit sources $X_1, X_2, \ldots, X_t$ with min-entropy bounds $ \Hmin(X_i)  \geq k_i$ for every $i \in [t].$
\end{definition}
Multi-source extractors are functions that transform multiple weak-sources to uniform randomness. Multi-source extractors have been studied extensively in the classical setting~\cite{Bou05,CG85,KLRZ08,R06,Raz05,KLR09}. With the advent of quantum computers, it is natural to investigate the security of extractors against a quantum adversary with quantum side information on weak-sources. 
As expected quantum side information presents many more challenges compared to classical side information. Gavinsky et al.~\cite{GKKRW07} gave an example of a seeded extractor secure against a classical adversary but not secure against a quantum adversary (even with a very small side information). Several definitions of quantum multi-source adversaries have been proposed in the literature. Kasher and Kempe~\cite{KK10} introduced quantum bounded storage adversary ($\qbsa$), where the adversary has bounded memory. They also introduced quantum independent adversary ($\qia$) which obtains independent side information from various sources. Arnon-Friedman, Portmann and Scholz~\cite{APS16} introduced quantum Markov adversary ($\qmra$), such that $X-E-Y$ forms a {\em Markov-chain} ($\condmutinf{X}{Y}{E}=0$)~\footnote{$\condmutinf{A}{B}{C}_\rho$ represents conditional mutual information between registers $A,B$ given register $C$ in state $\rho$.}, where $E$ is adversary's side information. Chung, Li and Wu~\cite{CLW14} introduced general entangled adversary ($\gea$).   A natural question that arises is what is the relationship/hierarchy between these different adversary models. Previous works~\cite{CLW14,KK10,APS16} also ask if there is a model which is stronger than the existing models. To quote~\cite{KK10} \begin{quote}
    In light of this, it is not clear if and how entangled guessing entropy sources can be incorporated into the model, and hence we only consider bounded storage adversaries in the entangled case.
\end{quote}A difficulty that they point is that in the quantum setting (unlike the classical setting), measuring the adversary side information might break the independence of the sources~\cite{KK10}. Quantum entanglement between various parties, used to generate side information raises additional issues. Entanglement is of course known to yield several unexpected effects with no classical counterparts, e.g., {\em non-local correlations}~\cite{Bell64} and {\em super-dense coding}~ \cite{CS92} etc. 

%

A main objective of this work is to unify all the existing quantum multi-source adversary models. We propose two new models of adversaries: 1) the quantum measurement adversary ($\qma$, Definition~\ref{qmadv}), which generates side information using entanglement and conditioned on post-measurement outcomes and 2) the quantum communication adversary ($\qca$, Definition~\ref{def:qcadv}), which generates side information using entanglement and communication between multiple sources.
\begin{definition}[$l\mhyphen\qma$,~$l\mhyphen\qmas$]\label{qmadv} 
	Let $\tau_{X\hat{X}}$, $\tau_{Y\hat{Y}}$ be canonical purifications of uniform sources $X,Y$ respectively (registers $\hat{X}\hat{Y}$ with Reference). 
	\begin{enumerate}
		\item  Alice and Bob hold $X,Y$ respectively. They also share an entangled state $\tau_{NM}$ (Alice holds $N$, Bob holds $M$).
		\item  Alice applies an (safe) isometry $V_A :  \cH_{X} \otimes \cH_{N}   \rightarrow   \cH_{X} \otimes \cH_{N'} \otimes \cH_{A}$ and Bob applies an (safe) isometry  $V_B :    \cH_Y \otimes \cH_{M} \rightarrow   \cH_{Y} \otimes \cH_{M'} \otimes \cH_{B}$. Registers $A, B$ are single qubit registers. Let $$\rho_{X\hat{X}AN'M'BY\hat{Y}} = (V_A \otimes V_B) (\tau_{X\hat{X}} \otimes \tau_{NM} \otimes \tau_{Y\hat{Y}})(V^\dagger_A \otimes V^\dagger_B).$$
		\item Alice and Bob perform a measurement in the  computational basis on the registers $A$ and $B$ respectively. Let 
		\[l = \log\left( \frac{1}{ \Pr(A=1, B=1)_{\rho}}  \right) \quad ; \quad \Phi_{X\hat{X}N'M'Y\hat{Y}} =\rho_{X\hat{X}AN'M'BY\hat{Y}} \vert (A=1,B=1). \]
		\item  Adversary gets access to either $\Phi_{XN'}$ or $\Phi_{M'Y}.$ The pure state $\Phi$ is called an $l\mhyphen\qmas$. 
	\end{enumerate}
\end{definition}

 \noindent \textbf{Motivation for the $l\mhyphen\qma$ model.}  A general model classically for two sources with side information is the Markov model of the form \[\mathcal{C} = \{ X-E-Y : \hmin{X}{E}  \geq k_1 \ \textnormal{and} \ \hmin{Y}{E}  \geq k_2\}.\]

We can generate Markov-chain as above via the following procedure. Let Alice and Bob hold independent and uniform sources $X'$ and $Y'$ respectively. In addition, let Alice and Bob share independent randomness $E'$. Let Alice generate a local random variable $A\in\{0,1\}$ (using $X'E'$) and Bob generate a local random variable $B\in\{0,1\}$ (using $Y'E'$). It can be readily verified that $(X'E'Y'~|~A=B=1)$ is a Markov-chain. It can also be checked that any Markov-chain  $X-E-Y$ can be generated using this procedure.  

The natural analogue of this procedure in the quantum setting is as described in Definition~\ref{qmadv}. A key difference with the classical setting is that, in the quantum setting, the resultant post-measurement state, does not form a Markov-chain and thus bringing the elusive nature of quantum side information. This brings us to consider sources~\footnote{Even though the state contain adversary side information, we call entire pure state as source for simplicity.} of the form 
\[\mathcal{Q} = \{ \sigma_{X\hat{X}N'M'Y\hat{Y}} : \sigma_{X\hat{X}N'M'Y\hat{Y}} \ \textnormal{is an} \ l\mhyphen\qmas\}.\]
A question may be asked if one can provide both the registers $M'$ and $N'$ as side information to the  adversary. However this may allow adversary to gain complete knowledge of $X,Y$ (since $N'$ may contain a copy of $X$ and $M'$ may contain a copy of $Y$) making the model trivial. Thus we settle on the model as in Definition~\ref{qmadv}.

We next define $2$-source quantum communication adversary inspired from quantum communication protocols.
\begin{definition}[$(k_1,k_2)\mhyphen\qca$,~$(k_1,k_2) \mhyphen \qcas$]~\label{def:qcadv}
	Let $\tau_{X\hat{X}}$, $\tau_{Y\hat{Y}}$ be the canonical purifications of the independent sources $X,Y$ respectively (registers $\hat{X}\hat{Y}$ with Reference). 
	\begin{enumerate}
	\item  Alice and Bob hold $X,Y$ respectively. They also share an entangled pure  state $\tau_{NM}$ (Alice holds $N$, Bob holds $M$).
	\item  Alice and Bob execute a quantum communication protocol, at the end of which the final state is $\ket{\Phi}_{XY\hat{X}\hat{Y}N'M'}$ (Alice holds $XN'$ and Bob holds $YM'$),
		with $\hmin{X}{M'Y\hat{Y}}_{\Phi} \geq k_1$ and 
		$\hmin{Y}{N'X\hat{X}}_{\Phi} \geq k_2$.
		\item   Adversary gets access to either one of  $\Phi_{XN'}$ or $\Phi_{YM'}$ of its choice. The state $\Phi$ is called a $(k_1,k_2) \mhyphen \qcas.$
	\end{enumerate}
\end{definition}
\begin{remark}
We note to the reader that indeed every $(k_1,k_2) \mhyphen \qcas$ is  a 
$(k_1,k_2)\mhyphen\qpas$ (where, $\qpas$ stands for quantum purified adversary state) defined in Boddu, Jain and Kapshikar~\cite{BJK21} (see Definition~\ref{qmadvk1k2}). But, it is apriori not clear if every $(k_1,k_2)\mhyphen\qpas$ can be generated via a communication protocol as in Definition~\ref{def:qcadv}. 
\end{remark}

\subsection*{Our results}
We show that,
  \begin{enumerate}
     \item $\qma$ is the strongest adversary among all the known adversaries (Theorem~\ref{thm2-intro1}), in the sense that the side information of all other adversaries can be generated by $\qma$.
       \item The (generalized) inner-product function (in fact a general class of two-wise independent functions) continues to work as a good extractor against $\qma$ (Theorem~\ref{corr:iphminhminintro223}) with matching parameters as that of Chor and Goldreich~\cite{CG85}  against classical adversaries. 
    \item A non-malleable extractor proposed by Li~\cite{Li12c} (against classical adversaries) continues to be secure against quantum side information (Theorem~\ref{refnmext_new}). A non-malleable extractor ($\nmext$) for sources ($X,Y$) is an extractor such that $\nmext (X,Y)$ is uniform and independent of $\nmext (X,Y')YY'$, where $Y'\ne Y$ is generated by the adversary using $Y$ and the side information on $X$. 
    
    This result implies a non-malleable extractor result of Aggarwal,  Chung,  Lin and  Vidick\newline~\cite{ACLV18} with uniform $Y$. We strengthen their result via a completely different proof to make the non-malleable extractor of  Li ($\nmext(X,Y) = \langle X, Y \vert \vert Y^2 \rangle$) secure against quantum side information even when $Y$ is not uniform.
     \item A modification (working with weak local randomness instead of uniform local randomness) of the Dodis and Wichs~\cite{DW09} protocol for privacy-amplification (PA) is secure against active quantum adversaries (those who arbitrarily modify the messages exchanged in the protocol). This strengthens on a recent result due to~\cite{ACLV18} which uses uniform local randomness.  
    \item We also show a tight efficiency lower bound (Corollary~\ref{refipbound_newintro}) for the (generalized) inner-product function (in fact a general class of two-wise independent functions). 
 \end{enumerate}


\subsection*{$\qma$ can simulate other adversaries}
We show that the side information of all the adversaries can be simulated~(see Definition~\ref{def:simulation}) in the model of $\qma$.

\begin{theorem}\label{thm2-intro1}
	Quantum side information of $(b_1, b_2)$-$\qbsa$ (see Definition~\ref{qbsadv}), $(k_1, k_2)$-$\qia$ (see Definition~\ref{qiadv}), $(k_1, k_2) $-$\gea$ (see Definition~\ref{geadv}),  acting on an $(n,k'_1,k'_2)$-source can be simulated by an $l\mhyphen\qma$ for some $l\leq 2\min \{b_1,b_2 \} +2n-k'_1-k'_2$, $l \leq 2n-k_1-k_2$, $l \leq 2n-k_1-k_2$ respectively.
	
	Quantum side information of  $(k_1, k_2) $-$\qmra$ (see Definition~\ref{madv}) and $(k_1, k_2) $-$\qca$ (see Definition~\ref{def:qcadv}) can be simulated $\eps$-approximately by an $l\mhyphen\qma$ for some $l \leq 2n-k_1-k_2+25+6 \log(1/\eps) $, $l \leq 2n-k_1-k_2+25+6 \log(1/\eps) $ respectively.
\end{theorem}

\subsection*{Inner-product is secure against $\qma$}
A $2$-source extractor secure against $l \mhyphen\qma$ is defined as follows:
\begin{definition}\label{qma2source}
An $(n,n,m)$-$2$-source extractor  $2\Ext : \{0,1\}^n \times \{0,1\}^n \to \{0,1\}^m$ is said to be $(l,\eps)$-quantum secure
against $l \mhyphen\qma$ if for every $l \mhyphen \qmas$ $\Phi$ (chosen by $l \mhyphen\qma$), we have 
$$  \|\Phi_{2\Ext(X,Y)N'} - U_m \otimes \Phi_{N'} \|_1 \leq \eps \tab \text{and} \tab  \| \Phi_{2\Ext(X,Y)M'} - U_m \otimes \Phi_{M'} \|_1 \leq \eps.$$ The extractor is called $Y$-strong if
$$  \| \Phi_{2\Ext(X,Y)M'Y} - U_m \otimes \Phi_{M'Y} \|_1 \leq \eps, $$
and $X$-strong if
$$  \|\Phi_{2\Ext(X,Y)N'X} - U_m \otimes \Phi_{N'X} \|_1 \leq \eps. $$
\end{definition}

\suppress{
\begin{definition}\label{qma2sourcek1k2}
An $(n,n,m)$-$2$-source extractor  $2\Ext : \{0,1\}^n \times \{0,1\}^n \to \{0,1\}^m$ is said to be $(k_1,k_2,\eps)$-quantum secure
against $(k_1,k_2) \mhyphen\qpa$ if for every $(k_1,k_2) \mhyphen \qpas$~$\Phi_{X\hat{X}N'M'Y\hat{Y}}$ (chosen by $(k_1,k_2)\mhyphen\qpa$), we have 
$$  \|\Phi_{2\Ext(X,Y)N'} - U_m \otimes \Phi_{N'} \|_1 \leq \eps \tab \text{and} \tab  \| \Phi_{2\Ext(X,Y)M'} - U_m \otimes \Phi_{M'} \|_1 \leq \eps.$$ The extractor is called $Y$-strong if
$$  \| \Phi_{2\Ext(X,Y)M'Y} - U_m \otimes \Phi_{M'Y} \|_1 \leq \eps, $$
and $X$-strong if
$$  \|\Phi_{2\Ext(X,Y)N'X} - U_m \otimes \Phi_{N'X} \|_1 \leq \eps. $$
\end{definition}}
We show that the inner-product extractor (in fact a general class of $\X$-two-wise independent function (Definition~\ref{def:infoquant}~[\ref{2wisefunction}])) of Chor and Goldreich \cite{CG85} continues to be secure against $l \mhyphen\qma$.

\begin{theorem}\label{corr:iphminhminintro223}
Let $p=2^m$ and $n'= n \log p$.  Let $\rho_{X \hat{X} N Y \hat{Y} M}$ be an $l\mhyphen\qmas$ such that $\vert X \vert = \vert \hat{X} \vert= \vert Y \vert= \vert \hat{Y} \vert =n'$ and $XY$ classical (with copies $\hat{X}\hat{Y}$ respectively). Let $f : \X \times \Y \to \Z$ be a $\X $-two-wise independent function such that $\X  =  \Y = \mathbb{F}_p^{n}$, $\Z = \mathbb{F}_p$ and $(X,Y) \in (\X,\Y)$. Let $Z= f(X, Y)\in \Z$. Then, \[ \|\rho_{Z Y M} - U_m  \otimes \rho_{YM} \|_1 \leq \eps,  \]for parameters 
$l \leq \left(n'-m-40+8 \log \left( \eps \right)\right)/2.$

Symmetric results follow for a $\Y$-two-wise independent function $f : \X \times \Y \to \Z$ by exchanging $(N,X) \leftrightarrow (M,Y)$ above.
\end{theorem}
As a corollary of Theorem~\ref{corr:iphminhminintro223}, we also get a  tight efficiency (Definition~\ref{def:eff}) lower bound for a $\X$-two-wise independent function.
\begin{corollary}[Efficiency lower bound for a $\X$-two-wise independent function]\label{refipbound_newintro} Let function  $f : \X \times \Y \to \Z$ be a $\X$-two-wise independent function such that $\vert \X \vert = \vert \Y \vert$. Let $U$ be the uniform distribution on $\X \times \Y$.
		 For any $\gamma > 0$,
		 $$ \log \left( \eff_{\gamma}(f,U) \right) \geq  \frac{1}{2} \left(  \log \vert \X \vert - \log \vert \Z \vert -40+ 8 \log \left( 1-\gamma-\frac{1}{ \vert \Z \vert } \right) \right).$$
\end{corollary}

\suppress{
We provi
	Let the inputs $X \in \X$ and $Y \in \Y$ be given to Alice and Bob respectively according to distribution $U$. Consider an optimal zero-communication protocol $\Pi$ with error of protocol under $U$ on non-abort being $\gamma$. Let the state  shared between Alice and Bob after their local operations be $\tau_{XN'M'YAB}$. Let $\perp$ represent the abort symbol. Let $\Pr(A= \perp  \vee B= \perp)_\tau \leq 1- \eta,$ and  $$\Phi_{XN'M'YAB} = (\tau_{XN'M'YAB}|A \ne \perp \wedge B \ne \perp)\enspace,$$ where Alices holds $XN'A$ and Bob holds $M'YB$. We have,
\[ \Pr(B \ne f(X,Y))_{\Phi}  \le \gamma  \tab \implies \tab  \Pr(B = f(X,Y))_{\Phi}  \ge 1-\gamma  . \]
Let $\| \Phi_{BYM'}- U_{ \log \vert\Z\vert} \otimes \Phi_{YM'} \|_1 \defeq \eps $. This implies,   $1-\gamma \leq  \Pr(B = f(X,Y))_{\Phi}  \leq \frac{1}{\vert \Z \vert}+\eps$. Noting $A \ne \perp$ (here) as $A=1$ (in Definition~\ref{qmadv}), $B \ne \perp$ (here) as $B=1$ (in Definition~\ref{qmadv}), state $\Phi$ is an $l \mhyphen \qmas$ with $l =\log \left( \frac{1}{\Pr(A \ne \perp  \wedge B \ne \perp)_\tau}\right).$ Since, $\| \Phi_{BYM'}- U_{ \log \vert\Z\vert} \otimes \Phi_{YM'} \|_1 = \eps \geq 1-\gamma- \frac{1}{\vert\Z \vert} $, using Theorem~\ref{corr:iphminhminintro223} we have 
 $$ \log \left( \eff_{\gamma}(f,U) \right) \geq \log \left( \frac{1}{\Pr(A \ne \perp  \wedge B \ne \perp)_\tau}\right)   \geq  \frac{1}{2} \left( \log \vert \X \vert - \log \vert \Z \vert -40+8 \log \left( 1-\gamma-\frac{1}{ \vert \Z \vert } \right) \right)$$
 which gives the desired.
 \mycomment{added till here}}

 Noting entanglement-assisted communication complexity of a function $ f$ is lower bounded by efficiency of a function $f$ (Fact~\ref{qcclowereff})  and that the  (generalized) inner-product function is a $\X$-two-wise independent function (note $\X = \mathbb{F}^n_p$), we immediately get the following as a corollary. 

\begin{corollary}\label{ipsecuritycorr:intro}	Let $\IP^n_p:  \mathbb{F}^n_p \times\mathbb{F}^n_p \to \mathbb{F}_p$ be defined as, 
	$$\IP_p^n(x,y) = \sum_{i=1}^{n}x_iy_i \mod p\enspace.$$  We have,
	$$ \Q_{\gamma}(\IP_p^n) \geq \frac{(n-1) \log p}{4} + \log \left(1-\gamma - \frac{1}{p}\right) -20\enspace.$$
\end{corollary}

\subsection*{A quantum secure weak-seeded  non-malleable extractor}
Quantum secure non-malleable extractors~\cite{ACLV18} are studied  when the seed is completely uniform. That is for a state $\sigma_{XMY}$, we have  $\sigma_{XMY} =\sigma_{XM} \otimes \sigma_Y$,
 $\sigma_Y=U_{Y}$ and we have $\hmin{X}{M}_\sigma \geq k.$ In this work, we extend the definition to study when the seed is not uniform.

Let $\sigma_{XMY}$ be such that  $\sigma_{XMY} =\sigma_{XM} \otimes \sigma_Y$, $\hmin{X}{M}_\sigma \geq k_1$ and $\Hmin(Y)_\sigma \geq k_2$. One may consider register $Y$ as weak-seed and register $M$ as  adversary quantum side information on source $X$. We consider the pure state extension of $\sigma_{XMY}$ denoted by  $\sigma_{X\hat{X}NMY\hat{Y}} = \sigma_{X\hat{X}NM} \otimes \sigma_{Y\hat{Y}}$ and call it a $(k_1,k_2)\mhyphen\qpasw$. 

\begin{definition}[$(k_1,k_2)\mhyphen\qpasw$]\label{qmadvk1k2weak}
We call a pure state $\sigma_{X\hat{X}NMY\hat{Y}}$, with $(XY)$ classical and $(\hat{X}\hat{Y})$ copy of $(XY)$,  a  $(k_1,k_2)\mhyphen\qpasw$ iff 
\[\sigma_{X\hat{X}NMY\hat{Y}} =\sigma_{X\hat{X}NM} \otimes \sigma_{Y\hat{Y}} \quad ;\quad  \hmin{X}{M}_\sigma \geq k_1 \quad ; \quad \Hmin(Y)_\sigma \geq k_2.\]
\end{definition}
A non-malleable extractor ($\nmext$) is an extractor such that $\nmext(X,Y)$ is uniform and independent of $\nmext(X,Y')YY'M'$, where $Y'\ne Y$ is generated by the adversary using $Y$ and the side information on $X$, i.e. $M$. Without any loss of generality, we consider the adversary operation to be isometry (since one can consider Stinespring extension of a CPTP map as an isometry~(see Fact~\ref{fact:stinespring}) if the adversary operation is a CPTP map). This leads us to consider $(k_1,k_2)\mhyphen\nmasw$.
\begin{definition}[$(k_1,k_2)\mhyphen\nmasw$]\label{def:2source-qnmadversarydefweak}
     Let $\sigma_{X\hat{X}NMY\hat{Y}}$ be a $(k_1,k_2)\mhyphen\qpasw$. Let $V: \cH_Y \otimes \cH_M \rightarrow \cH_Y \otimes \cH_{Y'} \otimes  \cH_{\hat{Y}'} \otimes \cH_{M'}$ be an isometry  such that for $\rho = V\sigma V^\dagger,$ we have $Y'$ classical (with copy $\hat{Y}'$) and $\Pr(Y \ne Y^\prime)_\rho =1.$ 
     We call state $\rho$ a $(k_1,k_2)\mhyphen\nmasw$.
\end{definition}
Since we require the non-malleable extractor to extract from every $(k_1,k_2)\mhyphen\nmasw$, we phrase an adversary $\nma$ (short for quantum non-malleable adversary) to choose the $(k_1,k_2)\mhyphen\nmasw$.
\begin{definition}[Quantum secure weak-seeded  non-malleable extractor]\label{def:2nm}
		An $(n_1,n_2,m)$-non-malleable extractor $\nmext : \{0,1\}^{n_1} \times \{0,1\}^{n_2} \to \{0,1\}^m$ is $(k_1,k_2,\eps)$-secure against $\nma$ if for every $(k_1,k_2)\mhyphen\nmasw$ $\rho$ (chosen by the adversary $\nma$),
	$$  \Vert \rho_{ \nmext(X,Y)\nmext(X,Y^\prime) Y  Y^\prime M^\prime} - U_m \otimes \rho_{ \nmext(X ,Y^\prime) Y  Y^\prime M^\prime} \Vert_1 \leq \eps. $$
\end{definition}
Following is an inner-product based non-malleable extractor proposed by Li~\cite{Li12c}. 
\begin{definition}[\cite{Li12c}]\label{ipnme}
	 Let $p \ne 2$ be a prime and 
	 $n$ be an integer. Define $\nmext : \mathbb{F}^n_p \times  \mathbb{F}^{n/2}_p \to  \mathbb{F}_p$ given by $\nmext(X,Y) \defeq \langle X, Y \vert \vert Y^2 \rangle$, where $\vert \vert$ represents concatenation of strings and $Y^2$ is computed via multiplication in   $\mathbb{F}_{p^{n/2}}$.
\end{definition}
We show that the inner-product based non-malleable extractor proposed by Li \cite{Li12c} continues to be secure against quantum side information.

\begin{theorem}\label{refnmext_new}
		Let $p \ne 2$ be a prime, $n$ be an even integer and $\eps > 0$. The function  $\nmext(X,Y)$ is a $(k_1,k_2,\eps)$-quantum secure weak-seeded non-malleable extractor against $\nma$ for the parameters $k_1 + k_2 \geq (n+17) \log p +33+16 \log \left( \frac{1}{\eps}  \right)$.
	
\end{theorem}

\subsection*{Privacy-amplification with weak-sources} 
We study the problem of privacy-amplification~(PA)~\cite{BBR88,Mau92,BBCM95,MW97}. In this problem, two parties, Alice and Bob, share a weak secret $X$ (with $\hmin{X}{E}\geq k$, where $E$ is adversary Eve side information). Using $X$ and an insecure communication channel, Alice and Bob would like to securely agree on a secret key $R$ that is close to uniformly random to an active adversary Eve who may have full control over their communication channel. In all prior protocols including~\cite{ACLV18}, we assume that Alice and Bob have local access to uniform sources of randomness. In practice, uniform sources are hard to come by, and it is more reasonable to assume that Alice and Bob have only weak-sources of randomness. For this we make use of breakthrough result by Dodis and Wichs \cite{DW09}, who were first to show the existence of a two-round PA protocol with optimal (up to constant factors) entropy loss, for any initial min-entropy. We modify the protocol of  \cite{DW09}, to accommodate the non-uniform local randomness at Alice and Bob side. Based on our construction of a quantum secure weak-seeded non-malleable extractor  (Theorem~\ref{refnmext_new}) we  obtain a PA protocol working with weak local sources of randomness and is secure against active quantum adversaries as long as the initial secret $X$ has min-entropy rate of more than half.

\subsection*{Proof overview} 

\suppress{
For the proof of Theorem~\ref{thm2-intro1}, we first establish Claims~\ref{claim:measure1}~and~\ref{claim:measure2} using the standard quantum information-theoretic techniques. Claim~\ref{claim:measure1} corresponds to correlating independent and uniform $X$ with quantum side information while Claim~\ref{claim:measure2} corresponds to correlating independent and uniform $Y$ with quantum side information conditioned on post-measurement by Alice. In particular, we use the quantum analogue of rejection-sampling argument of~\cite{Jain:2009} to prove them. We then prove  Lemma~\ref{lemma:simeverything}~(using Claim~\ref{claim:measure2}). Using Lemma~\ref{lemma:simeverything}, we generate side information (of other adversary models) in the $l\mhyphen \qma$ model, since it suffices to showing appropriate conditional-min-entropy and modified-conditional-min-entropy bounds of various adversary models in the purification picture. }

For the proof of Theorem~\ref{thm2-intro1}, we first prove Lemma~\ref{lemma:simeverything}. To prove Lemma~\ref{lemma:simeverything}, we first establish Claims~\ref{claim:measure1}~and~\ref{claim:measure2} using the standard quantum information-theoretic techniques.  Claim~\ref{claim:measure1} corresponds to correlating independent and uniform $X$ with quantum side information while Claim~\ref{claim:measure2} corresponds to correlating independent and uniform $Y$ with quantum side information conditioned on post-measurement by Alice. In particular, we used a quantum analogue of the rejection-sampling argument of~\cite{Jain:2009} to prove them. Next, using Lemma~\ref{lemma:simeverything}, we generate side information (of other adversary models) in the $l\mhyphen \qma$ model, since it suffices to showing appropriate conditional-min-entropy and modified-conditional-min-entropy bounds of various adversary models in the purification picture.

For the proof of Theorem~\ref{modgame}, let initial state of $l\mhyphen \qma$ model be $\tau$, state after the Alice measurement outcome $A=1$ be $\rho$ and state after the Bob measurement outcome $B=1$ conditioned on Alice post-measurement outcome $A=1$ be $\Phi$. We first reduce the task of bounding $ \Vert \Phi_{ZYM} -U_Z \otimes \Phi_{YM} \Vert_1 $ to bounding the collision-probability $\Gamma (\rho_{ZYM} |\rho_{YM})$ at the expense of multiplicative factor given by $ \frac{\vert \supp(Z) \vert}{\Pr(B=1)}$ using Cauchy Schwarz inequality as key ingredient. Notice, in state $\rho$, we have $\rho_{YX\hat{X}N'}= \rho_Y \otimes \rho_{ X\hat{X}N'}$. Thus, this enables us to use the argument of Renner~\cite{Renner05} to further bound the collision-probability $\Gamma (\rho_{ZYM} |\rho_{YM})$ to be exponentially small in min-entropy $\hminn{X}{YM}_\rho$ for pairwise-independent function $f$ such that $Z=f(X,Y)$. Additionally when $\tau_M =\rho_M$, the proof follows by noting the relation of $\hminn{X}{YM}_\rho$ with probability of Alice's measurement outcome $A=1$ ($\Pr(A=1)$).

\suppress{
For the proof of Theorem~\ref{corr:iphminhminintro}, we use the arguments of~\cite{BJK21} involving Facts~\ref{lem:hmin_and_tilde_relation},~\ref{uhlmann}~and~\ref{fact:substate_perturbation} to relate $(k_1,k_2) \mhyphen \qpas$ to an $l \mhyphen \qmas$ (with $l \approx 2n'-k_1-k_2$). To conclude Theorem~\ref{corr:iphminhminintro} from Theorem~\ref{modgame}, we need such a relation of a $(k_1,k_2) \mhyphen \qpas$ to a particular type of $l \mhyphen \qmas$ (with $\tau_M=\rho_M$ in Theorem~\ref{modgame}). This is additional novelty over the arguments of~\cite{BJK21}.}

For the proof of Theorem~\ref{refnmext_new}, we first make use the result of~\cite{ACLV18} which reduces the task of showing non-malleable extractor security of Li's extractor to showing hardness of inner-product in a guessing game. We note that it is equivalent to showing 
hardness of inner-product in a $(k_1,k_2)\mhyphen\nmasw$ (see Definition~\ref{def:2source-qnmadversarydefweak}). We next show that we can get the conditional-min-entropy bounds required to make use of Fact~\ref{corr:iphminhminintro} to simulate 
$(k_1,k_2)\mhyphen\nmasw$ as $(k_1,k_2) \mhyphen\qpas$. Thus, hardness of inner-product in a guessing game further reduces to showing security of inner-product against $(k_1,k_2) \mhyphen\qpas$. Using Fact~\ref{corr:iphminhminintro}, the proof now follows.

\subsection*{Comparison with~\cite{ACLV18}}
Both~\cite{ACLV18}~and~Theorem~\ref{refnmext_new} have considered the inner-product based non-malleable extractor proposed by Li~\cite{Li12c}.~\cite{ACLV18} extends the first step of classical proof, the reduction provided by the non-uniform XOR lemma, to the quantum case. This helps in reducing the task of showing non-malleable extractor property of inner-product to showing security of inner-product in a certain communication game. They then approach the problem of showing security of inner-product in a communication game by using the “reconstruction paradigm” of~\cite{DPVR09} to guess the entire input $X$ from the modified side information.

On the other hand, in Theorem~\ref{refnmext_new}, we reduce the security of inner-product in a communication game to the security of inner-product against the quantum measurement adversary. In the process, both~\cite{ACLV18} and~Theorem~\ref{refnmext_new} crucially use the combinatorial properties of inner-product. For example, in the proof of Theorem~\ref{refnmext_new}, we heavily uses the pairwise independence property of inner-product.
	
\subsection*{Other related works} 
	Seeded extractors have been studied extensively in the classical setting~\cite{ILM89,FS04}. K\"{o}nig et al. \cite{KT08} showed that any one-bit output extractor is also secure against quantum adversaries, with roughly the same parameters. Ta-Shma~\cite{TA09}, De and Vidick \cite{DV10}, and later De et al. \cite{DPVR09} gave  seeded extractors  with short seeds that are secure against quantum side information and can extract almost all of min-entropy and are based on Trevisan’s extractor \cite{Trevisan01}. The extractor of Impagliazzo et al.~\cite{ILM89} was shown to be secure against quantum side information by~\cite{KMR05,Renner05,RK05}.

In the multi-source setting, a probabilistic argument shows the existence of $2$-source extractors for min-entropy $k =\log n+ O(1)$. Explicit constructions of multi-source extractors  with access to more than $2$-sources has been considered in the successful line of work~\cite{BIRW06,Li12c,Li11,R09,Li13} leading to a near optimal $3$-source extractor that works for polylogarithmic min-entropy and has negligible error~\cite{Li15}. Explicit constructions of $2$-source extractors has been first considered in \cite{CG85} who showed that inner-product is a $2$-source extractor for min-entropy $k \geq n/2$. After nearly two decades, Bourgain \cite{Bou05} broke the “half entropy barrier”, and constructed a $2$-source extractor for min-entropy $(1/2-\delta)n$, for some tiny constant $\delta > 0.$ A long line of research starting with~\cite{CG85,Bou05,Raz05,CZ19,Li19} leading to a near optimal $2$-source
extractor that works for polylogarithmic min-entropy and has inverse polynomial error~\cite{CZ19,Li19}.

\subsection*{Subsequent works}
Inspired from our work, Boddu, Jain and Kapshikar~\cite{BJK21} have defined $(k_1,k_2)\mhyphen\qpas$ as specified below. 
\begin{definition}[$(k_1,k_2)\mhyphen\qpas$]\label{qmadvk1k2}
A pure state $\sigma_{X\hat{X}NMY\hat{Y}}$, with $(XY)$ classical and $(\hat{X}\hat{Y})$ copy of $(XY)$,  a  $(k_1,k_2)\mhyphen\qpas$ iff 
\[ \hmin{X}{MY\hat{Y}}_\sigma \geq k_1 \quad ; \quad \hmin{Y}{NX\hat{X}}_\sigma \geq k_2.\]
\end{definition}
They showed that every $l\mhyphen\qmas$ is also a  $(k_1,k_2)\mhyphen\qpas$ as stated in the below fact.
\begin{fact}[\cite{BJK21}]\label{fact:lqmaisk1k2qma} Let $\sigma_{X\hat{X}NM Y \hat{Y}}$ be an $l\mhyphen\qmas$ such that $\vert X \vert = \vert \hat{X} \vert= \vert Y \vert= \vert \hat{Y} \vert =n$. There exists $k_1,k_2$ such that $\sigma$ is a $(k_1,k_2)\mhyphen\qpas$, $k_1 \geq n-l$ and $k_2 \geq n-l$.
\end{fact}

They also showed that for every $(k_1,k_2)\mhyphen\qpas$ there is a close-by $l\mhyphen\qmas$ as stated in the below fact.
\begin{fact}[\cite{BJK21}] \label{lemma:nearby_rho_prime_prime4} Let $\rho_{X \hat{X} N Y \hat{Y} M}$ be a $(k_1,k_2)\mhyphen\qpas$ such that $\vert X \vert = \vert \hat{X} \vert= \vert Y \vert= \vert \hat{Y} \vert =n$. There exists an  $l\mhyphen\qmas$ $\rho^{(1)}$  such that,
\[ \Delta_B(\rho^{(1)}; \rho) \leq 6\eps \quad and \quad l \leq 2n- k_1 - k_2+ 4 +  6\log \left( \frac{1}{\eps} \right). \]
Furthermore, 
\[ \hminn{X}{MY\hat{Y}}_{\rho^{(1)}} \geq k_1-2\log \left( \frac{1}{\eps} \right) \quad ; \quad \hminn{Y}{NX\hat{X}}_{\rho^{(1)}} \geq k_2 -4 -  4\log \left( \frac{1}{\eps} \right).\]

\end{fact}

They used $(k_1,k_2)\mhyphen\qpas$ framework to construct the first explicit quantum secure non-malleable extractor for (source) min-entropy $\geq  \mathsf{polylog} \left( \frac{n}{\eps} \right)$ and seed length of $\mathsf{polylog} \left( \frac{n}{\eps} \right)$ ($n$ is the length of the source and $\eps$ is the error parameter)
 which lead to a $2$-round privacy-amplification protocol that is secure against active quantum adversaries with communication $\mathsf{polylog} \left( \frac{n}{\eps} \right)$, exponentially improving upon the linear communication required by the protocol due to~\cite{ACLV18}.

In addition, using our result, the security of inner-product against $l\mhyphen\qma$ as key ingredient, they constructed the first explicit quantum secure $2$-source non-malleable extractor for min-entropy $k_1,k_2 \geq n- n^{\Omega(1)}$, with an output of size $n/4$~and error $2^{- n^{\Omega(1)}}$. 

Also, recently Aggarwal, Boddu and Jain~\cite{ABJ22} have extended the connection of~\cite{CG14b} between $2$-source non-malleable extractors and non-malleable codes in the split-state model in classical setting to quantum setting. They further used quantum secure $2$-source non-malleable extractor~\cite{BJK21} to construct the first explicit quantum secure non-malleable code in the split-state model for message length $m=n^{\Omega(1)}$, error $\eps=2^{-n^{\Omega(1)}}$ and codeword size $2n$.

Additionally, Jain and Kundu~\cite{JK21} have used our efficiency lower bound result (Corollary~\ref{refipbound_new}) to obtain a direct-product result for two-wise independent functions including for the generalized inner-product function ($\IP$).


\subsection*{Organization} 
In Section~\ref{sec:prelims}, we present our notations, definitions and other information-theoretic preliminaries. In Section~\ref{sec3:ip}, we present the proof of Theorem~\ref{modgame}~and~Fact~\ref{corr:iphminhminintro}.  In Section~\ref{sec:prevext}, we present the proof of Theorem~\ref{thm2-intro1}. In Section~\ref{sec:app} we present the proof of Theorem~\ref{refnmext_new}.

\section{Preliminaries}
\label{sec:prelims}
\subsection*{Quantum information theory}

Let $\X, \Y, \Z$ be finite sets (we only consider finite sets in this paper). We use $x \leftarrow \X$ to denote $x$ drawn uniformly from $\X$. All the logarithms are evaluated to the base $2$. Consider a finite dimensional Hilbert space $\cH$ endowed with an inner-product $\langle \cdot, \cdot \rangle$ (we only consider finite dimensional Hilbert-spaces). A quantum state (or a density matrix or a state) is a positive semi-definite matrix on $\cH$ with trace equal to $1$. It is called {\em pure} if and only if its rank is $1$.  Let $\ket{\psi}$ be a unit vector on $\cH$, that is $\langle \psi,\psi \rangle=1$.  With some abuse of notation, we use $\psi$ to represent the state and also the density matrix $\ketbra{\psi}$, associated with $\ket{\psi}$. Given a quantum state $\rho$ on $\cH$, {\em support of $\rho$}, called $\text{supp}(\rho)$ is the subspace of $\cH$ spanned by all eigenvectors of $\rho$ with non-zero eigenvalues.
 
A {\em quantum register} $A$ is associated with some Hilbert space $\cH_A$. Define $\vert A \vert := \log \dim(\cH_A)$. Let $\mathcal{L}(\cH_A)$ represent the set of all linear operators on $\cH_A$. For operators $O, O'\in \cL(\cH_A)$, the notation $O \leq O'$ represents the L\"{o}wner order, that is, $O'-O$ is a positive semi-definite matrix. We denote by $\mathcal{D}(\cH_A)$, the set of quantum states on the Hilbert space $\cH_A$. State $\rho$ with subscript $A$ indicates $\rho_A \in \mathcal{D}(\cH_A)$. If two registers $A,B$ are associated with the same Hilbert space, we shall represent the relation by $A\equiv B$. For two states $\rho_A, \sigma_B$, we let $\rho_A \equiv \sigma_B$ represent that they are identical as states, just in different registers. Composition of two registers $A$ and $B$, denoted $AB$, is associated with the Hilbert space $\cH_A \otimes \cH_B$.  For two quantum states $\rho\in \mathcal{D}(\cH_A)$ and $\sigma\in \mathcal{D}(\cH_B)$, $\rho\otimes\sigma \in \mathcal{D}(\cH_{AB})$ represents the tensor product ({\em Kronecker} product) of $\rho$ and $\sigma$. The identity operator on $\cH_A$ is denoted $\id_A$. Let $U_A$ denote maximally mixed state in $\cH_A$. We also use $U_m$ to denote uniform distribution supported on $m$-bit strings. Let $\rho_{AB} \in \mathcal{D}(\cH_{AB})$. Define
$$ \rho_{B} \defeq \tr_{A}{\rho_{AB}} \defeq \sum_i (\bra{i} \otimes \id_{B})
\rho_{AB} (\ket{i} \otimes \id_{B}) , $$
where $\{\ket{i}\}_i$ is an orthonormal basis for the Hilbert space $\cH_A$.
The state $\rho_B\in \mathcal{D}(\cH_B)$ is referred to as the marginal state of $\rho_{AB}$. Unless otherwise stated, a missing register from subscript in a state will represent partial trace over that register. Given $\rho_A\in\mathcal{D}(\cH_A)$, a {\em purification} of $\rho_A$ is a pure state $\rho_{AB}\in \mathcal{D}(\cH_{AB})$ such that $\tr_{B}{\rho_{AB}}=\rho_A$. Purification of a quantum state is not unique. Suppose $A\equiv B$. Given $\{\ket{i}_A\}$ and $\{\ket{i}_B\}$ as orthonormal bases over $\cH_A$ and $\cH_B$ respectively, the \textit{canonical purification} of a quantum state $\rho_A$ is $\ket{\rho_A} \defeq (\rho_A^{\frac{1}{2}}\otimes\id_B)\left(\sum_i\ket{i}_A\ket{i}_B\right)$. 

A quantum {map} $\cE: \mathcal{L}(\cH_A)\rightarrow \mathcal{L}(\cH_B)$ is a completely positive and trace preserving (CPTP) linear map (mapping states in $\mathcal{D}(\cH_A)$ to states in $\mathcal{D}(\cH_B)$). A {\em unitary} operator $V_A:\cH_A \rightarrow \cH_A$ is such that $V_A^{\dagger}V_A = V_A V_A^{\dagger} = \id_A$. The set of all unitary operators on $\cH_A$ is  denoted by $\mathcal{U}(\cH_A)$. An {\em isometry}  $V:\cH_A \rightarrow \cH_B$ is such that $V^{\dagger}V = \id_A$ and $VV^{\dagger} = \id_B$. A {\em POVM} element is an operator $0 \le M \le \id$. We use shorthand $\bar{M} \defeq \id - M$, where $\id$ is clear from the context. We use shorthand $M$ to represent $M \otimes \id$, where $\id$ is clear from the context.

\begin{definition}[Classical register in a pure state]\label{def:classicalinpurestate}Let $\X$ be a set. A {\em classical-quantum} (c-q) state $\rho_{XE}$ is of the form \[ \rho_{XE} =  \sum_{x \in \X}  p(x)\ket{x}\bra{x} \otimes \rho^x_E , \] where ${\rho^x_E}$ are states.

Let $\rho_{XEA}$ be a pure state. We call $X$ a classical register in $\rho_{XEA}$, if $\rho_{XE}$ (or $\rho_{XA}$) is a c-q state. We identify random variable $X$ with the register $X$, with $\Pr(X=x) =p(x)$.
\suppress{
In a pure state $\rho_{XEA}$ in which $\rho_{XE}$ (or $\rho_{XA}$) is c-q, we call $X$ a classical register and identify random variable $X$ with it with $\Pr(X=x) =p(x)$.}
\end{definition}

\begin{definition}[Copy of a classical  register]\label{def:copyofaclassicalregister}
Let $\rho_{X\hat{X}E}$ be a pure state with $X$ being a classical register in $\rho_{X\hat{X}E}$ (see Definition~\ref{def:classicalinpurestate}) taking values in $\cX$. Similarly, let $\hat{X}$ be a classical register in $\rho_{X\hat{X}E}$ taking values in $\cX$. Let $\Pi_{\mathsf{Eq}} = \sum_{x \in \cX} \ketbra{x} \otimes \ketbra{x}$ be the \emph{equality} projector acting on the registers $X\hat{X}$. We call $X$ and $\hat{X}$ copies of each other (in the computational basis) if $\tr\left(\Pi_{\mathsf{Eq}} \rho_{X\hat{X}}\right) =1$.
\end{definition}

\begin{definition}[Conditioning] \label{def:conditioning}
Let  
\[ \rho_{XE} =  \sum_{x \in \{0,1\}^n}  p(x)\ket{x}\bra{x} \otimes \rho^x_E , \]
be a c-q state. For an event $\mathcal{S} \subseteq \{0,1\}^n$, define  $$\Pr(\mathcal{S})_\rho \defeq  \sum_{x \in \mathcal{S}} p(x) \quad ; \quad (\rho|X\in \mathcal{S})\defeq \frac{1}{\Pr(\mathcal{S})_\rho} \sum_{x \in \mathcal{S}} p(x)\ket{x}\bra{x} \otimes \rho^x_E.$$
We sometimes shorthand $(\rho|X\in \mathcal{S})$ as $(\rho|\mathcal{S})$ when the register $X$ is clear from the context. 

Let $\rho_{AB}$ be a state with $|A|=n$. We define 
$(\rho|A \in \mathcal{S}) \defeq (\sigma|\mathcal{S})$, where $\sigma_{AB}$ is the c-q state obtained by measuring the register $A$ in $\rho_{AB}$ in the computational basis. In case $\mathcal{S}=\{s\}$ is a singleton set, we shorthand $(\rho|A = s) \defeq \tr_A (\rho|A =s)$.
\end{definition}

\begin{definition}[Extension] \label{def:extension} Let $$\rho_{XE}=  \sum\limits_{x \in \{0,1\}^n}  p(x)\ket{x}\bra{x} \otimes \rho^x_E,$$
be a c-q state. For a function $Z:\cX \rightarrow \cZ$, define the following extension of $\rho_{XE}$, 
\[ \rho_{ZXE} \defeq  \sum_{x\in \cX}  p(x) \ket{Z(x)}\bra{Z(x)} \otimes \ket{x}\bra{x} \otimes  \rho^{x}_E.\]
\end{definition} 

\begin{definition}[Safe maps] \label{def:safe}
We call an isometry $V: \cH_X \otimes \cH_A \rightarrow \cH_X \otimes \cH_B$, {\em safe} on $X$ iff there is a collection of isometries $V_x: \cH_A\rightarrow \cH_B$ such that the following holds.  For all states $\ket{\psi}_{XA} = \sum_x \alpha_x \ket{x}_X \ket{\psi^x}_A$,
$$V  \ket{\psi}_{XA} =  \sum_x \alpha_x \ket{x}_X V_x \ket{\psi^x}_A.$$
We call a CPTP map $\Phi: \mathcal{L}( \cH_X \otimes \cH_A) \rightarrow \mathcal{L}(\cH_X \otimes \cH_B)$, {\em safe} on classical register $X$ iff there is a collection of CPTP maps $\Phi_x: \mathcal{L}(\cH_A)\rightarrow \mathcal{L}(\cH_B)$ such that the following holds.  For all c-q states $\rho_{XA} = \sum_x \Pr(X=x)_{\rho} \ketbra{x} \otimes  \rho^x_A$,
$$\Phi({\rho}_{XA}) =  \sum_x \Pr(X=x)_{\rho} \ketbra{x} \otimes \Phi_x( \rho^x_A).$$
\end{definition}
All isometries (or in general CPTP maps) considered in this paper are safe on classical registers that they act on. CPTP maps applied by adversaries can be assumed w.l.o.g as safe on classical registers, by the adversary first making a (safe) copy of classical registers and then proceeding as before. This does not reduce the power of the adversary.

\suppress{

A {\em classical-quantum} (c-q) state $\rho_{XE}$ (with $X$ a classical random variable) is of the form \[ \rho_{XE} =  \sum_{x \in \cX}  P_X(x)\ket{x}\bra{x} \otimes \rho^x_E , \] where ${\rho^x_E}$ are states and $P_X(x) \defeq \Pr(X=x)_\rho$. For an event $E \subseteq \cX$, define  $$\Pr(E)_\rho =  \sum_{x \in E} P_X(x) \quad ; \quad (\rho|E)\defeq \frac{1}{\Pr(E)_\rho} \sum_{x \in E} P_X(x)\ket{x}\bra{x} \otimes \rho^x_E.$$  
For a function $Z:\cX \rightarrow \cZ$, define \[ \rho_{ZXE} \defeq  \sum_{x\in \cX}  P_{X}(x) \ket{Z(x)}\bra{Z(x)} \otimes \ket{x}\bra{x} \otimes  \rho^{x}_E .  \]

All isometries considered in this paper are safe on classical registers that they act on. Isometries applied by adversaries can be assumed w.l.o.g as safe on classical registers, by the adversary first making a (safe) copy of classical registers and then proceeding as before. This does not reduce the power of the adversary. }

\begin{definition}~\label{def:infoquant}    
\begin{enumerate}
\item For $p \geq 1$ and matrix $A$,  let $\| A \|_p$ denote the {\em Schatten} $p$-norm.  
\item Let $\Delta(\rho ; \sigma) \defeq \frac{1}{2} \|\rho - \sigma\|_1$.
We write $\approx_\eps$ to denote $\Delta(\rho ; \sigma) \le \eps$. 
\item Let $d(X)_\rho \defeq \Delta(\rho_X;U_X)$ and  $d(X|Y )_\rho \defeq \Delta(\rho_{XY}; U_X \otimes \rho_Y)$.

\item  {\bf $\X$-two-wise independent function:}\label{2wisefunction} 
We call a function  $f : \X \times \Y \to \Z$,  $\X$-two-wise independent iff for any two distinct $x_1, x_2 \in \X$, 
$$ \Pr_{ y \leftarrow \Y }( f(x_1,y) =f(x_2,y)) = \frac{1}{\vert \Z \vert}.$$ 

\item {\bf Fidelity:}  For states $\rho,\sigma: ~\F(\rho;\sigma)\defeq\|\sqrt{\rho}\sqrt{\sigma}\|_1.$

\item {\bf Bures metric:}  For states $\rho,\sigma: \Delta_B(\rho;\sigma)\defeq \sqrt{1-\F(\rho;\sigma)}.$ Being a metric, it satisfies the triangle inequality. 

\item {\bf Max-divergence~\cite{Datta09, Jain:2009}:}  For states $\rho,\sigma$ such that $\supp(\rho) \subset \supp(\sigma)$, $$ \dmax{\rho}{\sigma} \defeq  \min\{ \lambda \in \mathbb{R} :   \rho  \leq 2^{\lambda} \sigma \}.$$ 
\item {\bf Sandwitched-R{\'e}nyi-divergence:} Let $1 \neq \alpha > 0$. For states $\rho,\sigma$ such that $\text{supp}(\rho) \subset \text{supp}(\sigma)$,
$$\relentalpha{\rho}{\sigma}{\alpha} \defeq \frac{1}{\alpha -1} \log \tr \left(\sigma^{\frac{1-\alpha}{2\alpha}}\rho\sigma^{\frac{1-\alpha}{2\alpha}}\right)^\alpha . $$
\item {\bf Collision-probability~\cite{Renner05}:} For state $\rho_{AB}$,
$$ \Gamma (\rho_{AB} |\rho_B) \defeq  \tr \left( \rho_{AB} (\id_A \otimes \rho_B^{-1/2}) \right)^2 =  2^{\relentalpha{\rho_{AB}}{\id_A \otimes \rho_B}{2}} = \|(\id_A \otimes \rho_B^{-1/4})\rho_{AB} (\id_A \otimes \rho_B^{-1/4})\|^2_2.$$  
\item {\bf Min-entropy:}  For a random variable $X$, 

$$\Hmin(X)=\min_{x \in     \text{supp}(X)       } \log\left(\frac{1}{P_X(x)} \right).$$

 \item {\bf Conditional-min-entropy:}\label{def:condminentropy} For state $\rho_{XE}$, min-entropy of $X$ conditioned on $E$ is defined as,
 $$ \hmin{X}{E}_\rho = - \inf_{\sigma_E \in  \mathcal{D}(\cH_{E}) } \dmax{\rho_{XE}}{\id_X \otimes \sigma_E}    .$$
 
 \item {\bf Max-information~\cite{Datta09}:}\label{def:maxinfo}
  For a state $\rho_{AB}$, $$ \imax(A:B)_{\rho} \defeq   \inf_{\sigma_{B}\in \mathcal{D}(\cH_B)}\dmax{\rho_{AB}}{\rho_{A}\otimes\sigma_{B}} .$$
  If $\rho$ is a classical state (diagonal in the computational basis) then the $\inf$ above is achieved by a classical state $\sigma_B$.

 \suppress{
 \item {\bf Smooth conditional-min-entropy:}  For state $\rho_{XE}$, smooth-min-entropy of $X$ conditioned on $E$ is defined as,
 $$ \hmineps{X}{E}{\eps}_\rho =  \sup_{\sigma_{XE} : \|\sigma_{XE} -\rho_{XE} \|_1 \leq \eps  }  \hmin{X}{E}_\sigma  .$$}
 
\item {\bf Modified-conditional-min-entropy:}\label{def:condminentropymod}  For state $\rho_{XE}$, the modified-min-entropy of $X$ conditioned on $E$ is defined as,
$$ \hminn{X}{E}_\rho = - \dmax{\rho_{XE}}{\id_X \otimes \rho_E}    .$$

%
\item {\bf Markov-chain:}\label{cqcmarkov}  A state $\rho_{XEY}$ forms a Markov-chain (denoted $(X-E-Y)_\rho$) iff $\condmutinf{X}{Y}{E}_\rho=0$. 
\end{enumerate}
\end{definition}

For the facts stated below without citation, we refer the reader to  standard text books~\cite{NielsenC00,WatrousQI,Wil12,Wat16}.
%
%
%
\begin{fact}\label{fact:lp}
    For even $p: ~ \| A \|^p_p  = \tr(A^\dagger A)^{p/2}.$
\end{fact}
\begin{fact}[\cite{BJL21T}]
\label{rhoablessthanrhoaidentity}
 
For a c-q state  $\rho_{XA}$ ($X$ classical): $\rho_{XA}  \le \id_X \otimes \rho_{A}$ and   $\rho_{XA}  \le \rho_X \otimes \id_{A}$.

\end{fact}
\begin{fact}[\cite{BJL21T}]\label{fact:boundnew}
    Let $\rho_{XBD}$ be a c-q state ($X$ classical) such that $\rho_{XB} = \rho_X \otimes \rho_B$. Then, $$ \imax(X:BD)_\rho \leq  2|D|.$$
\end{fact} 
\begin{fact}
	\label{fact101}  
	For a quantum state $\rho_{XE}:~ \hmin{X}{E}_\rho  \geq \hminn{X}{E}_\rho  .$
\end{fact}

\begin{fact}[\cite{BJK21}]\label{lem:hmin_and_tilde_relation} Let $\rho \in \mathcal{D}(\cH_{AB})$.
There exists $\rho^\prime \in \mathcal{D}(\cH_{AB})$ such that 
\[ \Delta_B(\rho; \rho^\prime) \leq \eps \quad ; \quad \hminn{A}{B}_{\rho} \leq \hmin{A}{B}_{\rho} \leq \hminn{A}{B}_{\rho^\prime} + 2\log\left( \frac{1}{\eps}\right).\]

\end{fact}

\begin{fact}[\cite{CLW14}]
	\label{fact102}  
	 Let $\Phi:\mathcal{L} (\cH_M ) \rightarrow   \mathcal{L}(\cH_{M'} )$ be a CPTP map and let $\sigma_{XM'} =(\id \otimes \Phi)(\rho_{XM})$. Then $$\hmin{X}{M'}_\sigma  \geq \hmin{X}{M}_\rho.$$
\end{fact}

\suppress{
\begin{fact}[\cite{KT08}]
\label{fact2}  
$ \hmin{X}{E}_\rho  \geq \Hmin(X) - \log(\dim(\cH_{E}) )   .$
\end{fact}}

\begin{fact}[]
	\label{tracedis}
	Let $\rho, \sigma , \tau$  be quantum states. Then 
		 $ \Delta ( \rho  ; \sigma)  \le \Delta (\rho  ; \tau) +\Delta (\sigma  ; \tau).$    
\end{fact}

\begin{fact}[\cite{FvdG06}]
\label{fidelty_trace}
Let $\rho,\sigma$ be two states. Then,
\[  1-\F(\rho;\sigma) \leq \Delta(\rho ; \sigma) \leq \sqrt{ 1-\F^2(\rho;\sigma)} \quad ; \quad \Delta_B^2(\rho;\sigma) \leq \Delta(\rho ; \sigma) \leq  \sqrt{2}\Delta_B(\rho;\sigma).  \]

\end{fact}
\begin{fact}[Data-processing]
\label{data}
Let $\rho, \sigma$  be quantum states and $\Phi$ be a CPTP map. Then 
\begin{itemize}
    \item $ \Delta ( \Phi(\rho)  ; \Phi(\sigma))  \le \Delta (\rho  ; \sigma).$   
     \item $ \Delta_B ( \Phi(\rho)  ; \Phi(\sigma))  \le \Delta_B (\rho  ; \sigma).$    
    \item  $\dmax{ \Phi(\rho) }{ \Phi(\sigma) }  \le \dmax{\rho}{ \sigma} .$  
    \item (\cite{Beigi13}) For all $1 \neq \alpha > 0: ~\relentalpha{ \Phi(\rho) }{ \Phi(\sigma) }{\alpha}  \le \relentalpha{\rho}{ \sigma}{\alpha} .$  
\end{itemize}
\end{fact}

\begin{fact}[\cite{H04}]\label{fact:markov}
    A Markov-chain $(X-E-Y)_\rho$ can be decomposed as follows: $$\rho_{XEY} = \sum_{t} \Pr(T=t) \ketbra{t} \otimes \left(\rho^t_{XE_1} \otimes \rho^t_{YE_2} \right),$$ where $T$ is some classical register over a  finite alphabet. 
\end{fact}
\begin{fact}[\cite{anshu2021oneshot}]\label{fact:markov2}
    For a Markov-chain $(X-E-Y)_\rho$, there exists a CPTP map $\Phi:\mathcal{L}( \cH_E ) \rightarrow \mathcal{L}( \cH_E \otimes \cH_Y)$ such that $\rho_{XEY} =({\id_X} \otimes \Phi) \rho_{XE}$.
\end{fact}  
\begin{fact}[Hölder’s inequality]
	\label{holders2}
	For matrices $A, B$, for any real $p,q >0$ and $\frac{1}{p} + \frac{1}{q} =1$ we have $ \vert \tr(A^\dagger B) \vert \le   \| A \|_p   \| B\|_q.$
\end{fact} 

\begin{fact}[Hölder’s inequality, special case]
\label{holders}
For matrices $A, B: ~ \| BAB^\dagger \|_1 \le   \| A \|_2   \| B\|^2_4  .$
\end{fact}
\begin{fact}[Cyclicity of trace]
	\label{cyctrace}
	For matrices $A, B, C: ~ \tr(ABC) = \tr(BCA) = \tr(CAB) .$
\end{fact}

\begin{fact}
\label{measuredmax}
Let $\rho_{AB} \in \mathcal{D}(\cH_A \otimes \cH_B)$ be a state and $M \in \cL(\cH_B)$ such that $M^\dagger M \leq \id_B$. Let $\hat{\rho}_{AB}= \frac{M \rho_{AB} M^\dagger}{\tr{M \rho_{AB} M^\dagger}}$. Then, 
$$\dmax{\hat{\rho}_A}{\rho_A} \leq \log \left(\frac{1}{\tr{M \rho_{AB} M^\dagger}}\right).$$
\end{fact}

\suppress{
\mycomment{added below fact}
\begin{fact}[Substate Perturbation Lemma~\cite{JK21}] \label{fact:substate_perturbation}
Let $\sigma_{XB}$,  $\psi_X$ and $\rho_B$ be states such that~\footnote{The statement in~\cite{JK21} is more general and is stated for {\em purified distance}, however it holds for any fidelity based distance including the Bures metric.} ,
\[  \quad   \sigma_{XB} \leq 2^c \left( \psi_X \otimes \sigma_{B}\right) \quad ; \quad \Delta_B\left(\sigma_B; \rho_B\right) \leq \delta_1  .\] 
For any $\delta_0 > 0$, there exists state $\rho^\prime_{XB}$ satisfying
\[\Delta_B\left(\rho^\prime_{XB}; \sigma_{XB} \right)\leq \delta_0 +\delta_1 \quad and \quad \rho^\prime_{XB} \leq 2^{c+1} \left(1+ \frac{4}{\delta_0^2}  \right) \left(\psi_X \otimes \rho_B\right). \]
\end{fact}
}
\begin{fact}[Stinespring isometry extension~\cite{WatrousQI}]\label{fact:stinespring}
		 Let $\Phi :    \mathcal{L} (\cH_X ) \rightarrow   \mathcal{L}(\cH_Y )$ be a CPTP map. There exists an isometry $V :  \cH_{X} \rightarrow   \cH_{Y} \otimes \cH_{Z}$ (Stinespring isometry extension of $\Phi$) such that $\Phi(\rho_X)= \tr_{Z}(V \rho_X V^\dagger)$ for every state $\rho_X$.
\end{fact} 

\begin{fact}[Corollary 5.5 in \cite{WatrousQI}]
	\label{measurediso}
	Let $\rho_{AB} \in \mathcal{D}(\cH_A \otimes \cH_B)$ be a pure state and $V_B : \mathcal{L} (\cH_{B}) \rightarrow   \mathcal{L}(\cH_{B'} \otimes \cH_{C})$ be an isometry such that $\vert C \vert =1$. Let   $\sigma_{AB'C} = (\id_A \otimes V_B) \rho_{AB} (\id_A \otimes V_B)^\dagger $ and  $\Phi_{AB'} = (\sigma_{AB'C} \vert C=1)$. There exists an operator $M_B$  such that $0 \le  M^\dagger_{B}M_{B} \le \id_{B}$ and $$\Phi_{AB'} =\frac{(\id_A \otimes M_{B} ) \rho_{AB} (\id_A \otimes M_{B} )^\dagger }{\tr(\id_A \otimes M_{B} ) \rho_{AB} (\id_A \otimes M_{B} )^\dagger} \quad ; \quad \Pr[C=1] = \tr M_B \rho_B M_B^\dagger.$$
\end{fact} 

\begin{fact}[Lemma 5.4.3~\cite{Renner05}]
	\label{renato2wise}
	Let $\rho_{XM}$ be a c-q state ($X$ classical). Let $Z =f( X) $, where $f \leftarrow \mathcal{F}$ is drawn from a two-wise independent hash function family $ \mathcal{F}$. Let $\sigma_M \in \mathcal{D}(\cH_M)$ be any state. Then,
	\begin{align*}
		\mathbb{E}_{f \leftarrow  \mathcal{F}} \left[   \tr \left( (\id_Z  \otimes \sigma_M^{-1/2}) (\rho_{f(X)M} - U_Z \otimes \rho_M)  \right)^2 \right]  \le  \tr \rho_{XM}(\id_X \otimes \sigma_M^{-1/2})    \rho_{XM}   (\id_X \otimes \sigma_M^{-1/2}).
	\end{align*}
\end{fact}

\begin{fact}[\cite{Renner05}]
\label{renato}
Let $p$ be a prime number and $n$ be a positive integer. Let $\rho_{XM}$ be a c-q state ($X$ classical) with $\rho_X\in  \mathbb{F}^n_p$. \suppress{Let $\sigma_M \in \mathcal{D}(\cH_M)$ be a state that achieves the infimum in Definition~\ref{def:infoquant}~[\ref{def:condminentropy}] for state $\rho_{XM}$.}Let $Z =f(X) $, where $f \leftarrow \mathcal{F}$ is drawn from a two-wise independent hash function family $ \mathcal{F}$. Then,

\begin{align*}
\mathbb{E}_{f \leftarrow \mathcal{F}} \left[   \tr \left( (\id_Z  \otimes \rho_M^{-1/2}) (\rho_{f(X)M} - U_Z \otimes \rho_M)  \right)^2 \right]  \le 2^{- \hminn{X}{M}_\rho}.
\end{align*}

\suppress{

Let $\rho_{XMY} = \rho_{XM} \otimes U_Y$ be a c-q state ($XY$ classical) with $\rho_X\in  \mathbb{F}^n_p$  and $U_Y\in  \mathbb{F}^n_p$ being uniform. Let $Z = \langle X, Y \rangle$. Then,
\begin{align*}
\mathbb{E}_{y \leftarrow U_Y} \left[   \tr \left( (\id_Z  \otimes \rho_M^{-1/2}) (\rho^y_{ZM} - U_Z \otimes \rho_M)  \right)^2 \right]  \le 2^{- \hmin{X}{M}_\rho}.
\end{align*}}
\end{fact}

\begin{proof}
Consider,
\begin{align*}
&\mathbb{E}_{f \leftarrow \mathcal{F}} \left[   \tr \left( (\id_Z  \otimes \rho_M^{-1/2}) (\rho_{f(X)M} - U_Z \otimes \rho_M)  \right)^2 \right]  \\ 
&\le   \tr \rho_{XM}(\id_X \otimes \rho_M^{-1/2})    \rho_{XM}   (\id_X \otimes \rho_M^{-1/2}) & \mbox{(Fact~\ref{renato2wise})}\\
&\le  \|(\id_X \otimes \rho_M^{-1/2})    \rho_{XM}   (\id_X \otimes \rho_M^{-1/2}) \|_\infty & \mbox{(Fact~\ref{holders2})}\\
&= 2^{- \hminn{X}{M}_\rho}. & \mbox{(Definition~\ref{def:infoquant}~[\ref{def:condminentropymod}])}
\end{align*}
\end{proof}

%
\begin{fact}[Uhlmann's theorem~\cite{uhlmann76}]
\label{uhlmann}
Let $\rho_A,\sigma_A\in \mathcal{D}(\cH_A)$. Let $\rho_{AB}\in \mathcal{D}(\cH_{AB})$ be a purification of $\rho_A$ and $\sigma_{AC}\in\mathcal{D}(\cH_{AC})$ be a purification of $\sigma_A$. There exists an isometry $V$ (from a subspace of $\cH_C$ to a subspace of $\cH_B$) such that,
 $$\F(\ketbra{\theta}_{AB}; \ketbra{\rho}_{AB}) = \F(\rho_A;\sigma_A) ,$$
 where $\ket{\theta}_{AB} = (\id_A \otimes V) \ket{\sigma}_{AC}$.
\end{fact}

\begin{fact}[Rejection-sampling~\cite{Jain:2009}]
	\label{fact:crejectionsampling}
	Let $X, Y$ be random variables such that $\dmax{X}{Y} \leq k$. There exits a random variable $Z \in \{0,1\}$, correlated with $Y$ such that $X  \equiv (Y | Z=1)$ and $\Pr(Z=1) \geq 2^{-k}$.
\end{fact} 

\begin{fact}[\cite{Jain:2009}]
\label{fact:rejectionsampling}
Let $\rho_{A'B}, \sigma_{AB}$ be pure states such that $\dmax{\rho_B}{\sigma_B} \leq k$. Let Alice and Bob share $\sigma_{AB}$. There exists an isometry $V: A \rightarrow A'C$ such that,
\begin{enumerate}
\item  $(V \otimes \id_B) \sigma_{AB}(V \otimes \id_B)^\dagger  = \phi_{A'BC}$, where $C$ is a single qubit register. 
\item Let $C$ be the outcome of measuring $\phi_C$ in the standard basis. Then $\Pr(C=1) \geq 2^{-k}$.
\item Conditioned on outcome $C=1$, the state shared between Alice and Bob is $\rho_{A'B}$.  
\end{enumerate}
\end{fact}
We present a proof here for completeness.
\begin{proof}
Since $\dmax{\rho_B}{\sigma_B} \leq k$, let $\sigma_{B} = 2^{-k} \rho_{B} + (1 - 2^{-k}) \tau_{B}$. Let $\tau_{A'B}$ be a purification of $\tau_B$. Consider,
$$ \phi_{CA'B} =  \sqrt{(1 - 2^{-k})} \ket{0}_C \otimes \tau_{A'B} +\sqrt{2^{-k}}\ket{1}_C \otimes \rho_{A'B} .$$
Notice $\phi_{A'BC} $ is purification of $\sigma_{B}$. From Fact~\ref{uhlmann}, there exists an isometry $V: A \rightarrow A'C$ such that $(V \otimes \id_B) \sigma_{AB} (V \otimes \id_B)^\dagger = \phi_{A'BC}.$ The desired properties can be readily verified. 
 \end{proof}

\begin{fact}[\cite{BJK21}]\label{fact:minentropydecrease_on_measuremen}
Let $\rho_{ABC} \in \mathcal{D}(\cH_A \otimes \cH_B \otimes \cH_C)$ be a state and $M \in \cL(\cH_C)$ such that $M^\dagger M \leq \id_C$. Let $\hat{\rho}_{ABC}= \frac{M \rho_{ABC} M^\dagger}{\tr{M \rho_{ABC} M^\dagger}}$. Then, 
\[ \hmin{A}{B}_{\hat{\rho}} \geq   \hmin{A}{B}_{\rho} - \log \left(\frac{1}{\tr{M \rho_{ABC} M^\dagger}}\right). \] 
Furthermore~\footnote{The proof in~\cite{BJK21} can be easily modified to prove the inequality for $\hminn{.}{.}$.} if $\hat{\rho}_B =\rho_B$, we also get, 
\[ \hminn{A}{B}_{\hat{\rho}} \geq   \hminn{A}{B}_{\rho} - \log \left(\frac{1}{\tr{M \rho_{ABC} M^\dagger}}\right). \]
\end{fact} 
\begin{fact}[\cite{BJK21}] \label{claim:traingle_rho_rho_prime}Let $\rho_{ZA}, \rho'_{ZA}$ be states such that $\Delta(\rho; \rho') \leq \eps'$. If $d(Z\vert A)_{\rho'} \leq \eps$ then $d(Z\vert A)_{\rho} \leq 2\eps' + \eps$.~\footnote{Claim holds even when $\Delta_B()$ is replaced with $\Delta()$.}
\end{fact}
\begin{claim}
	\label{claim:measure1}
Let $\phi_{XX'AB}$ be a pure state  such that $\Hmin({X \vert B})_\phi  \geq k$. Let $X$ be a classical register (with copy $X'$). Let $\theta_{X_1X_2}$ be the canonical purification of $\theta_{X_2}$ such that $\theta_{X_2} \equiv U_X$. Let $\theta_{X_1X_2}$ be shared between Reference ($X_1$) and Alice ($X_2$). There exists a pure state $\sigma_{AB}$ such that when shared between Alice ($A$) and Bob ($B$), Alice can perform a measurement which succeeds with probability at least $2^{k- \vert X \vert }$ and on success  joint shared state is $\phi_{X_1X'AB}$ between Reference ($X_1$), Alice ($X'A$) and Bob ($B$) such that  $\phi_{X_1X'AB} \equiv \phi_{XX'AB}$.
\end{claim}

\begin{proof}
	Since $\Hmin({X \vert B})_\phi  \geq k$, we have 
	$$ \inf_{\sigma_{B}} \dmax{\phi_{XB}}{U_{X} \otimes \sigma_{B}} \leq \vert X \vert -k.$$ 
	Let the infimum above be achieved by $\sigma_{B}$ and  let  $\sigma_{AB}$ be its purification shared between Alice ($A$) and Bob ($B$). The desired now follows from Fact~\ref{fact:rejectionsampling} by treating Bob (in Fact~\ref{fact:rejectionsampling}) as Reference and Bob (here), state $\sigma_{AB}$ (in Fact~\ref{fact:rejectionsampling}) as $\theta_{X_1X_2} \otimes \sigma_{AB}$ (here) and state $\rho_{AB}$ (in Fact~\ref{fact:rejectionsampling}) as $\phi_{X_1X'AB}$ (here).
\end{proof}

\begin{claim}
	\label{claim:measure2}
	Let $\phi_{XX'AYY'B}$ be a pure state such that 
	$$\Hmin({X \vert BYY'})_\phi  \geq k_1 \quad \text{     and     } \quad \hminn{Y}{XX'A}_\phi \geq k_2 .$$ Let $X, Y$ be  classical registers (with copy $X'$ and $Y'$ respectively). Let $\theta_{X_1X_2}$ be a  canonical purification of $\theta_{X_2}$ such that $\theta_{X_2} \equiv U_X$. Let $\theta_{X_1X_2}$ be shared between Reference ($X_1$) and Alice ($X_2$). Let $\tau_{Y_1Y_2}$ be a  canonical purification of $\tau_{Y_2}$ such that $\tau_{Y_2} \equiv U_Y$.
	Let $\tau_{Y_1Y_2}$ be shared between Reference ($Y_1$) and Bob ($Y_2$). There exists a pure state $\sigma_{ABYY'}$ such that when shared between Alice ($A$) and Bob ($BYY'$), Alice and Bob can each perform a measurement which jointly succeeds with probability at least $2^{k_1+k_2- \vert X \vert- \vert Y \vert }$ and on success  the joint shared state is $\phi_{X_1X'AY_1Y'B}$ between Reference ($X_1Y_1$), Alice ($X'A$) and Bob ($Y'B$), such that $\phi_{X_1X'AY_1Y'B} \equiv \phi_{XX'AYY'B}$. 
\end{claim}
\begin{proof}
Since $\Hmin({X \vert BYY'})_\phi  \geq k_1$, we have 
	$$ \inf_{\sigma_{BYY'}} \dmax{\phi_{XBYY'}}{U_{X} \otimes \sigma_{BYY'}} \leq \vert X \vert -k_1.$$ Let the infimum above be achieved by $\sigma_{BYY'}$ and  let  $\sigma_{ABYY'}$ be its purification shared between Alice ($A$) and Bob ($BYY'$). From Claim~\ref{claim:measure1}, Alice can do a measurement such that on success $\phi_{X_1X'AYY'B}$ is shared between Reference ($X_1$), Alice ($X'A$) and Bob ($YY'B$) (such that $\phi_{X_1X'AYY'B} \equiv \phi_{XX'AYY'B}$). Also, since $\hminn{Y}{ AXX'}_\phi  \geq k_2$, we have 
	$$ \dmax{\phi_{YAX_1X'}}{U_{Y} \otimes \phi_{AX_1X'}} \leq \vert Y \vert -k_2.$$ 
	Again from Claim~\ref{claim:measure1}, Bob can do a measurement such that on success $\phi_{X_1X'AY_1Y'B}$ is shared between Reference ($X_1Y_1$), Alice ($X'A$) and Bob ($Y'B$) (such that $\phi_{X_1X'AY_1Y'B} \equiv \phi_{XX'AYY'B}$). This completes the proof by noting probability of success in the first and second steps as $2^{k_1- \vert X \vert}$ and $2^{k_2- \vert Y \vert }$ respectively.
\end{proof}
\begin{claim}
	\label{claim:measure3}
	Let $\phi_{XX'AYY'B}$ be a pure state such that 
	$$\hminn{X}{BYY'}_\phi  \geq k_1 \quad \text{     and     } \quad \hminn{Y}{XX'A}_\phi \geq k_2 .$$ Let $X, Y$ be  classical registers (with copy $X'$ and $Y'$ respectively). Let $\theta_{X_1X_2}$ be a  canonical purification of $\theta_{X_2}$ such that $\theta_{X_2} \equiv U_X$. Let $\theta_{X_1X_2}$ be shared between Reference ($X_1$) and Alice ($X_2$). Let $\tau_{Y_1Y_2}$ be a  canonical purification of $\tau_{Y_2}$ such that $\tau_{Y_2} \equiv U_Y$.
	Let $\tau_{Y_1Y_2}$ be shared between Reference ($Y_1$) and Bob ($Y_2$). There exists a pure state $\sigma_{ABYY'}$ such that $\sigma_{BYY'} \equiv \phi_{BYY'}$ and  when shared between Alice ($A$) and Bob ($BYY'$), Alice and Bob can each perform a measurement which jointly succeeds with probability at least $2^{k_1+k_2- \vert X \vert- \vert Y \vert }$ and on success  the joint shared state is $\phi_{X_1X'AY_1Y'B}$ between Reference ($X_1Y_1$), Alice ($X'A$) and Bob ($Y'B$), such that $\phi_{X_1X'AY_1Y'B} \equiv \phi_{XX'AYY'B}$. 
\end{claim}
\begin{proof}

The proof follows in similar lines of Claim~\ref{claim:measure2} after a simple modification. We state the modification required and do not repeat the entire argument. Considering $\sigma_{ABYY'}$ to be any purification of $\sigma_{BYY'}$ such that $\sigma_{BYY'} \equiv \phi_{BYY'}$ in Claim~\ref{claim:measure2} and repeating the argument of Claim~\ref{claim:measure2}, the result follows.
\suppress{
    Simulation of the state $\phi$ in the model of $l\mhyphen\qma$ follows after a simple modification of Claim~\ref{claim:measure2}. We state the modification required. Claim~\ref{claim:measure2} holds even when $\hmin{}{}$ bound requirement is replaced by $\hminn{}{}$ bound requirement and considering $\sigma_{ABYY'}$ to be any purification of $\sigma_{BYY'}$ such that $\sigma_{BYY'} = \phi_{BYY'}$ in Claim~\ref{claim:measure2}.}
\end{proof}

\begin{lemma}\label{lemma:simeverything} Let $\rho_{X \hat{X} N Y \hat{Y} M}$ be a pure state such that $\vert X \vert = \vert \hat{X} \vert= \vert Y \vert= \vert \hat{Y} \vert =n$, $XY$ classical (with copies $\hat{X}\hat{Y}$ respectively) and 
\[\hmin{X}{Y\hat{Y}M}_\rho \geq k_1  \quad ; \quad \hminn{Y}{X\hat{X}N}_\rho \geq k_2. \]
Then $\rho$ is also an $l\mhyphen\qmas$ (see Definition~\ref{qmadv}) for some $l \leq 2n- k_1 - k_2.$
\end{lemma}
\begin{proof}
Simulation of the state $\rho$ in the model of $l\mhyphen\qma$ follows from Claim~\ref{claim:measure2}. Using Claim~\ref{claim:measure2},  with the following assignment of registers (below the registers on the left are from Claim~\ref{claim:measure2} and the registers on the right are the registers in this proof)
 $$(X, Y, X', Y', A, B, \Phi) \leftarrow (X, Y, \hat{X}, \hat{Y},  N, M, \rho),$$we have $$l  \leq  \log (2^{n-k_1+n-k_2}) = 2n-k_1-k_2.$$
\end{proof}

\begin{corollary}\label{lem:weakqma}
Let $\rho$ be a  $(k_1,k_2)\mhyphen\qpasw$ such that $\vert X \vert =\vert Y \vert =n$. Then $\rho$ is also an $l \mhyphen \qmas$ for the parameters $l \leq 2n-k_1-k_2$. 
\end{corollary}
\begin{proof}
    Note $\rho_{X \hat{X}NMY\hat{Y}} = \rho_{X \hat{X}NM} \otimes \rho_{Y\hat{Y}}$, where $\rho_{X \hat{X}NM}$ is the purification of $\rho_{XM}$ such that $\hmin{X}{M}_\rho \geq k_1$ and $\Hmin(Y)_\rho \geq k_2$. Since $ \rho_{X\hat{X}NMY\hat{Y}}  = \rho_{X\hat{X}NM}  \otimes \rho_{\hat{Y}Y},$ we have  $$\hmin{X}{Y\hat{Y}M}_\rho \geq k_1$$ and   
\[  \dmax{\rho_{YX\hat{X}N}  }{U_Y \otimes \rho_{X\hat{X}N} }  \leq  \log ( \dim(\cH_{Y}) ) -k_2=n-k_2.\] The first inequality follows since 
$\rho_{YX\hat{X}N}   = \rho_Y \otimes \rho_{X\hat{X}N}  $ and $\Hmin(Y)_\rho \geq k_2.$ Thus, $\hminn{Y}{X\hat{X}N}_\rho \geq k_2$. Simulation of the $(k_1, k_2)$-$\qpasw$ in the model of $l\mhyphen\qma$ now follows from
Lemma~\ref{lemma:simeverything}.
\suppress{
Claim~\ref{claim:measure2}. From Claim~\ref{claim:measure2},  in the simulation of $\rho_{X\hat{X}NMY\hat{Y}}$, in the model of $l\mhyphen\qma$,  we have $$l =  \log \left( \frac{1}{\Pr(A=1,B=1)}\right) \leq  \log (2^{n-k_1+n-k_2}) = 2n-k_1-k_2.$$
    }
\end{proof}

\subsection*{Extractors} 

\begin{definition}[Quantum secure seeded extractor]
\label{qseeded}
	An $(n,d,m)$-seeded extractor $\Ext : \{0,1\}^n \times \{0,1\}^d \to \{0,1\}^m$  is said to be $(k,\eps)$-quantum secure if for every state $\rho_{XES}$, such that $\Hmin(X|E)_\rho \geq k$ and $\rho_{XES} = \rho_{XE} \otimes U_d$, we have 
	$$  \| \rho_{\Ext(X,S)E} - U_m \otimes \rho_{E} \|_1 \leq \eps.$$
	In addition, the extractor is called strong if $$  \| \rho_{\Ext(X,S)SE} - U_m \otimes U_d \otimes \rho_{E} \|_1 \leq \eps .$$
	$S$ is referred to as the {\em seed} for the extractor.
	\end{definition}

\begin{fact}[\cite{DPVR09}~\cite{CV16}]
	\label{trevext}
	There exists an explicit $(2m,\eps)$-quantum secure strong $(n,d,m)$-seeded extractor $\Ext : \{ 0,1\}^n \times  \{ 0,1\}^d \to  \{ 0,1\}^m$ for parameters $d = O( \log^2(n/\eps) \log m )$. 
\end{fact}

\begin{definition}[$2$-source extractor against an adversary~\cite{KK10,CLW14,APS16}]\label{iadv2source}
An $(n,n,m)$-$2$-source extractor $2\Ext : \{0,1\}^n \times \{0,1\}^n \to \{0,1\}^m$ is $(\cdot,\eps)$-quantum~\footnote{$"\cdot"$ in $(\cdot ,\eps)$ indicates the parameters of adversary model used to define $\rho_{XEY}$. For example, $(l,\eps)$- quantum secure against $\qma$ (for $l\mhyphen\qmas$'s chosen by $\qma$).} secure against an adversary if for every $\rho_{XYE}$ (generated appropriately in the adversary model, adversary holds $E$), we have
$$  \| \rho_{2\Ext(X,Y)E} - U_m \otimes \rho_{E} \|_1 \leq \eps. $$ 
The extractor is called $Y$-strong if $$  \| \rho_{2\Ext(X,Y)EY} - U_m \otimes \rho_{EY} \|_1 \leq \eps,$$
and $X$-strong if $$  \| \rho_{2\Ext(X,Y)EX} - U_m \otimes \rho_{EX} \|_1 \leq \eps.$$
\end{definition}

\begin{definition}\label{def:mac}
	A function $\mac:\{0,1\}^{2m} \times\{0,1\}^m \to \{0,1\}^m$ is an \emph{$\eps$-information-theoretically secure one-time message authentication code} if for any function $\mathcal{A}:\{0,1\}^m \times \{0,1\}^m \to \{0,1\}^m\times \{0,1\}^m$ it holds that for all $\mu \in \{0,1\}^m$
	$$\Pr_{k\leftarrow \{0,1\}^{2m}}\big[ (\mac(k,\mu') = \sigma' ) \, \wedge \, (\mu'\neq \mu) : (\mu',\sigma') \leftarrow \mathcal{A}(\mu,\mac(k,\mu))\big] \,\leq\,\eps.$$  
\end{definition}
Efficient constructions of $\mac$ satisfying the conditions of Definition~\ref{def:mac} are known. The following fact summarizes some parameters that are achievable using a construction based on polynomial evaluation.
\begin{fact}[Proposition 1 in~\cite{KR09}]\label{prop:mac}
	For any integer $m > 0$, there exists an efficient family of  $2^{-m}$-information-theoretically secure one-time message authentication codes $$\mac:\{0,1\}^{2m} \times\{0,1\}^m \to \{0,1\}^m.$$
\end{fact}

\subsection*{Quantum communication complexity}
\label{sec:comm}
In a quantum communication protocol $\Pi$ for computing a function $f : \X \times \Y \to \Z$,  the inputs $x\in \X$ and $y\in\Y$ are given to Alice and Bob respectively. They also start with an entangled pure state independent of the inputs. They perform local operations and exchange quantum messages. The goal is to minimize the communication between them. Please refer to preliminaries of~\cite{Yao93, CB97} for details of an entanglement-assisted quantum communication protocol. Let $O(x,y)$ refer to the output of the protocol, on input $(x,y)$. Let $C_i$ (on $c_i$ qubits) refer to the $i$-th message sent in the protocol. 
\begin{definition}
    \label{def:qcc} Define,
\begin{align*}
	\textrm{Worst-case error:} & \quad \err(\Pi) \defeq \max_{(x,y)} \{ \Pr[O(x,y) \ne f(x,y)] \} .\\
	\textrm{Communication of a quantum protocol:}&\quad  \QCC(\Pi) \defeq \sum_i  c_i . \\
	\textrm{Entanglement-assisted communication complexity of $f$:}&\quad  \Q_\gamma(f) \defeq \min_{\Pi: \err(\Pi) \leq \gamma} \QCC(\Pi).  \\
\end{align*}
\end{definition}
%
%
\subsubsection*{Protocols with abort and efficiency}
\label{sec:eff}
 Consider the following zero-communication protocol with abort for a function $f :  \X \times \Y   \to \Z$. Let the inputs $X \in \X$ and $Y\in\Y$ be given to Alice and Bob respectively according to distribution $\mu$. They also start with an entangled pure state $\tau_{NM}$ independent of the inputs (Alice holds $N$ and Bob holds $M$). They apply local operations and measurements and are allowed to abort with some probability. Let the state  shared between Alice and Bob after their local operations be $\tau_{XN'M'YAB}$. Let $\perp$ represent the abort symbol. Let $\Pr(A= \perp  \vee B= \perp)_\tau \leq 1- \eta,$ and  $$\Phi_{XN'M'YAB} = (\tau_{XN'M'YAB}|A \ne \perp \wedge B \ne \perp)\enspace,$$ where Alices holds $XN'A$ and Bob holds $M'YB$. Let  $\gamma >0$. The goal of Alice and Bob is to maximize $\eta$ such that $$ \Pr(B \ne f(X,Y))_{\Phi}  \le \gamma.$$  
 \begin{definition}
     \label{def:eff} 
     Define:
\begin{align*}
\textrm{Error of $\Pi$ under $\mu$ on non-abort:}&\quad  \err(\Pi,\mu) \defeq  \Pr(B \ne f(X,Y))_{\Phi}. \\
 	\textrm{Efficiency of $\Pi$ under $\mu$:}&\quad  \eff(\Pi, \mu) \defeq \frac{1}{\eta} \\
	\textrm{Efficiency of $f$ under $\mu$:}&\quad \eff_\gamma(f,\mu) \defeq \min_{\Pi :  \err(\Pi,\mu) \leq \gamma} \eff(\Pi,\mu) \\
		\textrm{Efficiency of $f$:}&\quad  \eff_\gamma(f) \defeq  \max_{\mu}\eff_\gamma(f,\mu).
\end{align*}
 \end{definition}

\begin{fact}[Theorem 5 in \cite{SVJ12}] 
	\label{qcclowereff}
	For a function $f :  \X \times \Y   \to \Z$ and $\gamma > 0$,
	$$ \Q_{\gamma}(f) \geq \frac{1}{2}  \log (\eff_{\gamma}(f)) .$$
\end{fact}

\section{Inner-product is secure against $\qma$}\label{sec3:ip}
We show that the inner-product extractor of Chor and Goldreich~\cite{CG85} is secure against $\qma$ (with nearly the same parameters as that of classical adversary). More generally we show any $\X$-two-wise independent function (Definition~\ref{def:infoquant}~[\ref{2wisefunction}]) continues to be secure against $l \mhyphen\qma$. We first show that the security of $\X $-two-wise independent function against $l \mhyphen \qma$.

\begin{theorem}
	\label{modgame} Let $f : \X \times \Y \to \Z$ be a $\X $-two-wise independent function such that $\vert \X \vert = \vert \Y \vert$. 
		 
	\begin{enumerate}
		\item Let $\tau = \tau_{XX_1} \otimes \tau_{NM} \otimes \tau_{YY_1}$, where $\tau_{XX_1}$ is the canonical purification of $\tau_{X}$ (maximally mixed in $\X$), $\tau_{YY_1}$ is canonical purification of $\tau_{Y}$ (maximally mixed in $\Y$)  and $\tau_{NM}$ is a pure state. 
		\item Let $V_A : \cH_{X_1} \otimes \cH_{N} \rightarrow   \cH_{X_1} \otimes \cH_{N'} \otimes \cH_{A}$ be a (safe) isometry. Let $\rho = (V_A \tau V^\dagger_A~|~A=1).$
		\item  Let  $V_B :     \cH_Y \otimes \cH_{M} \rightarrow   \cH_{Y} \otimes \cH_{M'} \otimes \cH_{B}$ be a (safe) isometry. Let $\Theta = V_B \rho V^\dagger_B$ and $\Phi= (\Theta~|~B=1)$. 
		\item Let $Z= f(X, Y)\in \Z$ and $\eps \defeq \| \Phi_{Z Y M'} - U_Z  \otimes \Phi_{YM'} \|_1 $.
	\end{enumerate}
	Then 
	$$\hminn{X}{M}_\rho -\log \vert \Z \vert  +  \log \left(\Pr(B=1)_\Theta \right) \leq   2 \log \frac{1}{\eps}.$$
	Additionally if $\tau_M= \rho_M$, we further have,
	$$ \log \vert \X \vert - \log \vert \Z \vert  +  \log \left(\Pr(A=1, B=1)_{(V_A \otimes V_B)\tau(V_A \otimes V_B)^\dagger} \right) \leq   2 \log \frac{1}{\eps}.$$Symmetric results follow for a $\Y$-two-wise independent function $f : \X \times \Y \to \Z$ by exchanging $(N,A,X) \leftrightarrow (M,B,Y)$ above. 
\end{theorem}
\begin{proof}
From Fact~\ref{measurediso}, there exists operator $C_{YM}$ such that $0 \leq   C^\dagger_{YM}C_{YM}  \leq \id_{YM}$ and,
	\[  	\Phi_{ZXYM'} =     \frac{C_{YM}   \rho_{ZXYM}C_{YM}^{\dagger} }{\tr C_{YM}   \rho_{ZXYM}C_{YM}^{\dagger} } \quad  ; \quad  \Pr(B=1)_\Theta = \tr C_{YM}   \rho_{YM}C_{YM}^{\dagger}.\]  
This implies,
\[  	\Phi_{ZYM'} =     \frac{ C_{YM}   \rho_{ZYM}C_{YM} ^{\dagger} }{\tr C_{YM}   \rho_{ZYM} C_{YM}^{\dagger} }. \]
\suppress{
Let $\sigma_M \in \mathcal{D}(\cH_M)$ be a state that achieves the infimum in Definition~\ref{def:infoquant}~[\ref{def:condminentropy}] for state $\rho_{XM}$. Let $\sigma_{YM} =U_Y \otimes \sigma_M$.} Define $C \defeq \id_Z \otimes C_{YM}$ and $D \defeq \id_Z \otimes \rho_{YM}^{1/4}$.
	Consider (below we use cyclicity of trace, Fact~\ref{cyctrace}, several times without mentioning),
	\begin{align*}
		\|  CD \|^4_4 &=   \tr  (\id_Z \otimes \rho_{YM}^{1/4} C^\dagger_{YM}C_{YM}\rho_{YM}^{1/4} )^2           & \mbox{(Fact~\ref{fact:lp})}\\
		&=   \vert \Z \vert  \cdot \tr  ( \rho_{YM}^{1/2} C^\dagger_{YM}C_{YM}\rho_{YM}^{1/2} C^\dagger_{YM}C_{YM})         \\
		&\le    \vert \Z \vert  \cdot \|\rho_{YM}^{1/2} C^\dagger_{YM}C_{YM} \rho_{YM}^{1/2} \|_1 \| C^\dagger_{YM}C_{YM}  \|_\infty  & \mbox{(Fact~\ref{holders2})}   \\ 
		&\le    \vert \Z \vert  \cdot \|\rho_{YM}^{1/2} C^\dagger_{YM}C_{YM} \rho_{YM}^{1/2} \|_1  & \mbox{($ \|  C^\dagger_{YM}C_{YM} \|_\infty \le 1 $)}  \\
				&=    \vert \Z \vert  \cdot \tr(\rho_{YM}^{1/2} C^\dagger_{YM}C_{YM} \rho_{YM}^{1/2})   & \mbox{($\rho_{YM}^{1/2} C^\dagger_{YM}C_{YM} \rho_{YM}^{1/2}\geq 0$)}  \\
		&=    \vert \Z \vert  \cdot \tr(C_{YM} \rho_{YM} C^\dagger_{YM})   \\
		&=     \vert \Z \vert  \cdot \Pr(B=1)_\Theta. 
	\end{align*}Consider,\begin{align*}
		& \| \Phi_{ZYM'} - U_Z  \otimes \Phi_{YM'}  \|^2_1 \\
		& = \frac{1}{(\Pr(B=1)_\Theta)^2} \| C(\rho_{ZYM} - U_Z  \otimes \rho_{YM} )C^\dagger \|^2_1       \\                                                                                                  
		& = \frac{1}{(\Pr(B=1)_\Theta)^2}  \|  C D D^{-1}(\rho_{ZYM} - U_Z  \otimes \rho_{YM}) D^{-1}DC^\dagger \|^2_1                                                                                                                            &                             \\
		& \le \frac{1}{(\Pr(B=1)_\Theta)^2}  \|  CD  \|^4_4    \ \| D^{-1} (\rho_{ZYM} - U_Z  \otimes \rho_{YM}) D^{-1} \|^2_2                                                                                                                    & \mbox{(Fact~\ref{holders})} \\
		& \le \frac{\vert \Z \vert  \cdot \Pr(B=1)_\Theta}{ (\Pr(B=1)_\Theta)^2}   \cdot  \| D^{-1} (\rho_{ZYM} - U_Z  \otimes \rho_{YM}) D^{-1} \|^2_2   \\
		& \le \frac{\vert \Z \vert }{\Pr(B=1)_\Theta} \cdot \|(\id_Z  \otimes \rho_{YM}^{-1/4})(\rho_{ZYM} - U_Z  \otimes  \rho_{YM})(\id_Z  \otimes \rho_{YM}^{-1/4}) \|^2_2                   \\
		& = \frac{\vert \Z \vert }{\Pr(B=1)_\Theta} \cdot \mathbb{E}_{y \leftarrow \rho_Y} \left[ \|(\id_Z  \otimes \rho_{M}^{-1/4})(\rho^y_{ZM} - U_Z  \otimes  \rho_{M})(\id_Z  \otimes \rho_{M}^{-1/4}) \|^2_2    \right]                    \\
		& =  \frac{\vert \Z \vert }{\Pr(B=1)_\Theta} \cdot \mathbb{E}_{y \leftarrow \rho_Y} \left[ \tr \left((\id_Z  \otimes \rho_{M}^{-1/2})(\rho^y_{ZM} - U_Z  \otimes  \rho_{M}) \right)^2 \right]                     &                             \\
		& \le \ \frac{\vert \Z \vert }{\Pr(B=1)_\Theta}   \cdot 2^{- \hminn{X}{M}_\rho}.                                                                                                                                   & \mbox{(Fact~\ref{renato})}
	\end{align*}
	We get the first result by taking $\log$ on both sides and rearranging terms. 
	
	Now, let $\rho_M= \tau_M$. Noting $\tau_M= (V_A \tau V^\dagger_A)_{M}$ and using Fact~\ref{fact:minentropydecrease_on_measuremen},  with the following assignment of terms (below the terms on the left are from Fact~\ref{fact:minentropydecrease_on_measuremen} and the terms on the right are from this proof)
 $$(\rho_{ABC}, \hat{\rho}_{AB}) \leftarrow ( (V_A \tau V^\dagger_A)_{XMA}, \rho_{XM}),$$we get $$\hminn{X}{M}_\rho \ge \hminn{X}{M}_{(V_A \tau V^\dagger_A)}+ \log(\Pr(A=1)_{(V_A \tau V^\dagger_A)}) = \log \vert \X \vert  + \log(\Pr(A=1)_{(V_A \tau V^\dagger_A)}).$$We get the second result now by noting, 
$$\Pr(A=1)_{(V_A \tau V^\dagger_A)} \cdot \Pr(B=1)_\Theta =  \Pr(A=1, B=1)_{(V_A \otimes V_B)(\tau)(V_A \otimes V_B)^\dagger}.$$

\end{proof}

\begin{lemma}\label{lemma:simeverything1}Let $p=2^m$ and $n'= n \log p$. Let $\rho_{X \hat{X} N Y \hat{Y} M}$ be a pure state such that $\vert X \vert = \vert \hat{X} \vert= \vert Y \vert= \vert \hat{Y} \vert =n'$, $XY$ classical (with copies $\hat{X}\hat{Y}$ respectively) and 
\[\hminn{X}{Y\hat{Y}M}_\rho \geq k_1  \quad ; \quad \hminn{Y}{X\hat{X}N}_\rho \geq k_2. \]
 Let $Z= \IP^{n}_p(X,Y)$. Then $\rho$ is also an $l\mhyphen\qmas$ (see Definition~\ref{qmadv}) for some $l \leq 2n'- k_1 - k_2.$ Furthermore,
 \[\|\rho_{Z X N} - U_m  \otimes \rho_{XN} \|_1 \leq \eps \quad ; \quad \|\rho_{Z Y M} - U_m  \otimes \rho_{YM} \|_1 \leq \eps, \]for parameters 
$k_1+k_2 \geq n'+m+2 \log \left( \frac{1}{\eps} \right).$
 
\end{lemma}
\begin{proof}Simulation of the state $\rho$ in the model of $l\mhyphen\qma$ follows from Claim~\ref{claim:measure3}. Using Claim~\ref{claim:measure3},  with the following assignment of registers (below the registers on the left are from Claim~\ref{claim:measure3} and the registers on the right are the registers in this proof)
 $$(X, Y, X', Y', A, B, \Phi) \leftarrow (X, Y, \hat{X}, \hat{Y},  N, M, \rho),$$we have $$l  \leq  \log (2^{n'-k_1+n'-k_2}) = 2n'-k_1-k_2.$$

 Let $X \in \X, Y \in \Y$. Note that $\IP^{n}_p$ is a $\X$-two-wise independent function. Further state $\rho_{X \hat{X}NY\hat{Y}M}$ is an $l\mhyphen\qmas$ with $(\tau_M,\rho_M)$ in Theorem~\ref{modgame} equivalent to  $(\rho_{Y\hat{Y}M},\rho_{Y\hat{Y}M})$ in the simulation of the state $\rho$ using Claim~\ref{claim:measure3}. 
Let $\eps' \defeq 2\Delta( \rho_{Z Y M} \; ; \; U_Z  \otimes \rho_{YM}) $. Using Theorem~\ref{modgame},  
we get $\eps' \leq 2^{\frac{-(k_1+k_2-n'-m )}{2} }$. For the choice of  parameters 
$k_1+k_2 \geq n'+m+2 \log \left( \frac{1}{\eps} \right),$ we get $\eps' \leq \eps.$ Thus, 
$\|\rho_{Z Y M} - U_Z  \otimes \rho_{YM} \|_1 \leq \eps.$ Symmetric result follows noting $\IP^{n}_p$ is also a $\Y$-two-wise independent function.

\end{proof}

Using Lemma~\ref{lemma:simeverything1} as a key ingredient,~\cite{BJK21} showed that inner-product extractor (in fact a general class of $\X$-two-wise independent function (Definition~\ref{def:infoquant}~[\ref{2wisefunction}])) of Chor and Goldreich~\cite{CG85} is secure against $\qpas$. We state the result from~\cite{BJK21} as follows. 
\begin{fact}[\cite{BJK21}]\label{corr:iphminhminintro}
Let $p=2^m$ and $n'= n \log p$. Let $\rho_{X \hat{X} N Y \hat{Y} M}$ be a $(k_1,k_2)\mhyphen\qpas$ such that $\vert X \vert = \vert \hat{X} \vert= \vert Y \vert= \vert \hat{Y} \vert =n'$ and $XY$ classical (with copies $\hat{X}\hat{Y}$ respectively). Let $f : \X \times \Y \to \Z$ be a $\X $-two-wise independent function such that $\X  =  \Y = \mathbb{F}_p^{n}$, $\Z = \mathbb{F}_p$ and $(X,Y) \in (\X,\Y)$. Let $Z= f(X, Y)\in \Z$. Then, \[\|\rho_{Z Y M} - U_m  \otimes \rho_{YM} \|_1 \leq \eps,  \]for parameters 
$k_1+k_2 \geq n'+m+40+8 \log \left( \frac{1}{\eps} \right).$ 

Symmetric results follow for a $\Y$-two-wise independent function $f : \X \times \Y \to \Z$ by exchanging $(N,X) \leftrightarrow (M,Y)$ above.
\end{fact}

\begin{theorem}\label{corr:iphminhminintro2}
Let $p=2^m$ and $n'= n \log p$.  Let $\rho_{X \hat{X} N Y \hat{Y} M}$ be an $l\mhyphen\qmas$ such that $\vert X \vert = \vert \hat{X} \vert= \vert Y \vert= \vert \hat{Y} \vert =n'$ and $XY$ classical (with copies $\hat{X}\hat{Y}$ respectively). Let $f : \X \times \Y \to \Z$ be a $\X $-two-wise independent function such that $\X  =  \Y = \mathbb{F}_p^{n}$, $\Z = \mathbb{F}_p$ and $(X,Y) \in (\X,\Y)$. Let $Z= f(X, Y)\in \Z$. Then, \[ \|\rho_{Z Y M} - U_m  \otimes \rho_{YM} \|_1 \leq \eps,  \]for parameters 
$l \leq \left(n'-m-40+8 \log \left( \eps \right)\right)/2.$

Symmetric results follow for a $\Y$-two-wise independent function $f : \X \times \Y \to \Z$ by exchanging $(N,X) \leftrightarrow (M,Y)$ above.
\end{theorem}
\begin{proof}
    The proof follows from Fact~\ref{corr:iphminhminintro} after noting Fact~\ref{fact:lqmaisk1k2qma}.
\end{proof}

We also get a  tight efficiency (Definition~\ref{def:eff}) lower bound for a $\X$-two-wise independent function.
\begin{corollary}[Efficiency lower bound for a $\X$-two-wise independent function]\label{refipbound_new} Let function  $f : \X \times \Y \to \Z$ be a $\X$-two-wise independent function such that $\vert \X \vert = \vert \Y \vert$. Let $U$ be the uniform distribution on $\X \times \Y$.
		 For any $\gamma > 0$,
		 $$ \log \left( \eff_{\gamma}(f,U) \right) \geq  \frac{1}{2} \left(  \log \vert \X \vert - \log \vert \Z \vert -40+ 8 \log \left( 1-\gamma-\frac{1}{ \vert \Z \vert } \right) \right).$$
\end{corollary}
\begin{proof}
	Let the inputs $X \in \X$ and $Y \in \Y$ be given to Alice and Bob respectively according to distribution $U$. Consider an optimal zero-communication protocol $\Pi$ with error of protocol under $U$ on non-abort being $\gamma$. Let the state  shared between Alice and Bob after their local operations be $\tau_{XN'M'YAB}$. Let $\perp$ represent the abort symbol. Let $\Pr(A= \perp  \vee B= \perp)_\tau \leq 1- \eta,$ and  $$\Phi_{XN'M'YAB} = (\tau_{XN'M'YAB}|A \ne \perp \wedge B \ne \perp)\enspace,$$ where Alices holds $XN'A$ and Bob holds $M'YB$. We have,
\[ \Pr(B \ne f(X,Y))_{\Phi}  \le \gamma  \tab \implies \tab  \Pr(B = f(X,Y))_{\Phi}  \ge 1-\gamma  . \]
Let $\| \Phi_{BYM'}- U_{ \log \vert\Z\vert} \otimes \Phi_{YM'} \|_1 \defeq \eps $. This implies,   $1-\gamma \leq  \Pr(B = f(X,Y))_{\Phi}  \leq \frac{1}{\vert \Z \vert}+\eps$. Noting $A \ne \perp$ (here) as $A=1$ (in Definition~\ref{qmadv}), $B \ne \perp$ (here) as $B=1$ (in Definition~\ref{qmadv}), state $\Phi$ is an $l \mhyphen \qmas$ with $l =\log \left( \frac{1}{\Pr(A \ne \perp  \wedge B \ne \perp)_\tau}\right).$ Since, $\| \Phi_{BYM'}- U_{ \log \vert\Z\vert} \otimes \Phi_{YM'} \|_1 = \eps \geq 1-\gamma- \frac{1}{\vert\Z \vert} $, using Theorem~\ref{corr:iphminhminintro2} we have 
 $$ \log \left( \eff_{\gamma}(f,U) \right) \geq \log \left( \frac{1}{\Pr(A \ne \perp  \wedge B \ne \perp)_\tau}\right)   \geq  \frac{1}{2} \left( \log \vert \X \vert - \log \vert \Z \vert -40+8 \log \left( 1-\gamma-\frac{1}{ \vert \Z \vert } \right) \right)$$
 which gives the desired.

\end{proof}
	From Fact~\ref{qcclowereff}, and noting that the  (generalized) inner-product function is a $\X$-two-wise independent function (note $\X = \mathbb{F}^n_p$), we have,
\begin{corollary}\label{ipsecuritycorr:}	Let $\IP^n_p:  \mathbb{F}^n_p \times\mathbb{F}^n_p \to \mathbb{F}_p$ be defined as, 
	$$\IP_p^n(x,y) = \sum_{i=1}^{n}x_iy_i \mod p\enspace.$$  We have,
	$$ \Q_{\gamma}(\IP_p^n) \geq \frac{(n-1) \log p}{4} + \log \left(1-\gamma - \frac{1}{p}\right) -20\enspace.$$
\end{corollary}

\begin{corollary}\label{corr:iphminhmin}
Let $p=2^m$ and $n'= n \log p$.  Let $\rho_{X \hat{X} N Y \hat{Y} M}$ be a $(k_1,k_2)\mhyphen\qpas$ such that $\vert X \vert = \vert \hat{X} \vert= \vert Y \vert= \vert \hat{Y} \vert =n'$, $XY$ classical (with copies $\hat{X}\hat{Y}$ respectively). Let $Z= \IP^{n}_p(X,Y)$. Then, \[\|\rho_{Z X N} - U_m  \otimes \rho_{XN} \|_1 \leq \eps \quad ; \quad \|\rho_{Z Y M} - U_m  \otimes \rho_{YM} \|_1 \leq \eps,  \]for parameters 
$k_1+k_2 \geq n'+m+40+8 \log \left( \frac{1}{\eps} \right).$
\end{corollary}

\begin{proof}
Let $X \in \X, Y \in \Y$ such that $\X =\Y=\mathbb{F}_p^n$.
    The result follows from Fact~\ref{corr:iphminhminintro} and noting that $\IP^{n}_p$ is both $\X$-two-wise independent function and $\Y$-two-wise independent function.
\end{proof}

\begin{corollary}\label{corr:wqma}
Let $p=2^m$ and $n'= n \log p$. Let $\sigma_{XEY} = \sigma_{XE} \otimes \sigma_Y$ such that $\hmin{X}{E}_\sigma \geq k_1$,  $\Hmin(Y)_\sigma \geq k_2$ and $\vert X \vert=\vert Y \vert ={n'}$. Let $Z= \IP^{n}_p(X,Y)$. Then,
$$\Vert \sigma_{ZXE}-U_m \otimes \sigma_{XE} \Vert_1 \leq \eps,$$
for parameters $k_1+k_2   \geq  n'+m+40+ 8 \log \left(\frac{1}{\eps}\right).$
\end{corollary}
\begin{proof}
Consider
$\sigma_{X\hat{X}\hat{E} EY\hat{Y}} =\sigma_{X\hat{X}\hat{E} E} \otimes \sigma_{Y\hat{Y}}$, the purification of $\sigma_{XEY}$ such that $\sigma_{X\hat{X}\hat{E} E}$ is the canonical purification of $\sigma_{XE}$ and $\sigma_{Y\hat{Y}}$ is the canonical purification of $\sigma_{Y}$. Note $\sigma$ is a $(k_1,k_2)\mhyphen\qpas$. Now, the result follows from Corollary~\ref{corr:iphminhmin}.\suppress{Thus, from Lemma~\ref{lem:weakqma}, we have $\sigma$ is an $l \mhyphen \qmas$ for $l \leq 2n'-k_1-k_2.$

Note $\IP^{n}_p$ is both $\X$-two-wise independent function and $\Y$-two-wise independent function (note $\X =\Y= \mathbb{F}^n_p$). Let  $ \Vert \sigma_{ZX\hat{E}} -U_m \otimes \sigma_{X\hat{E}} \Vert_1 = \eps'$. From Theorem~\ref{thm1-intro}, we have $n'-m-l \leq 2 \log \left(\frac{1}{\eps'}\right)$.
     Since $l \leq 2n'-k_1-k_2,$ we further have $\eps' \leq 2^{\frac{n'+m-k_1-k_2}{2}}.$ For the choice of parameters $k_1+k_2   \geq  n'+m+ 2 \log \left(\frac{1}{\eps}\right),$ we get  $ \Vert \sigma_{ZX\hat{E}} -U_m \otimes \sigma_{X\hat{E}} \Vert_1 = \eps' \leq \eps.$ 
     Noting $\sigma_{XE} =\sigma_{X\hat{E}}$, we have the desired.}
\end{proof}

\section{$\qma$ can simulate other adversaries}\label{sec:prevext}

We show how $\qma$ can simulate all adversaries known in the literature and $\qca$ as well. By simulation we mean that the quantum side information of an adversary can be generated by $\qma$. More precisely, 
\begin{definition}[Simulation]\label{def:simulation}
    Let $\rho_{XEY}$ be the final state generated appropriately in the adversary model, adversary holds $\rho_E$. We say $\rho_{XEY}$ can be simulated by $\qma$ for parameter $l$, if there exists an $l \mhyphen \qmas,$ $\Phi_{X \hat{X}N MEY\hat{Y}}$~($\qma$ gets registers $EM$~\footnote{This amounts to giving more quantum side information ($ME$) than  other adversary model provide ($E$).}) and $\Phi_{XYE} =\rho_{XYE}$. Analogously, we say $\rho_{XEY}$ can be simulated $\eps \mhyphen$approximately by $\qma$ for parameter $l$, if there exists an $l \mhyphen \qmas,$ $\Phi_{X \hat{X}N MEY\hat{Y}}$ and $\Phi_{XYE}  \approx_{\eps}\rho_{XYE}$. 
\end{definition}

\begin{theorem}\label{thm2-intro}
	Quantum side information of $(b_1, b_2)$-$\qbsa$ (see Definition~\ref{qbsadv}), $(k_1, k_2)$-$\qia$ (see Definition~\ref{qiadv}), $(k_1, k_2) $-$\gea$ (see Definition~\ref{geadv}) acting on an $(n,k'_1,k'_2)$-source can be simulated by an $l\mhyphen\qma$ for some $l\leq 2\min \{b_1,b_2 \} +2n-k'_1-k'_2$, $l \leq 2n-k_1-k_2$, $l \leq 2n-k_1-k_2$ respectively.
	
	Quantum side information of  $(k_1, k_2) $-$\qmra$ (see Definition~\ref{madv}) and $(k_1, k_2) $-$\qca$ (see Definition~\ref{def:qcadv}) can be simulated $\eps$-approximately by an $l\mhyphen\qma$ for some $l \leq 2n-k_1-k_2+25+6 \log(1/\eps) $, $l \leq 2n-k_1-k_2+25+6 \log(1/\eps) $  respectively.
\end{theorem}
\suppress{
\begin{theorem}\label{thm2-intro}
		Quantum side information of any of $(b_1, b_2) $-$\qbsa$ (see Definition~\ref{qbsadv}), $(k'_1, k'_2)$-$\qia$ (see Definition~\ref{qiadv}), $(k'_1, k'_2) $-$\gea$ (see Definition~\ref{geadv}), $(k'_1, k'_2) $-$\qca$ (see Definition~\ref{def:qcadv}) acting on an $(n,k_1,k_2)$-source can be simulated by an $l\mhyphen\qma$ for some $l\leq \min \{b_1+2n-k_1-k_2,b_2+2n-k_1-k_2 \}$, $l \leq 2n-k'_1-k'_2$, $l \leq 2n-k'_1-k'_2$, $l \leq 2n-k'_1-k'_2$ respectively.
	
	Quantum side information of  $(k'_1, k'_2) $-$\qmra$ (see Definition~\ref{madv}) acting on an $(n,k_1,k_2)$-source can be simulated $\eps$-approximately by an $l\mhyphen\qma$ for some $l \leq 2n-k'_1-k'_2+16+4 \log(1/\eps) $ respectively.
\end{theorem}}
\begin{proof}
	The proof follows from Claims~\ref{sim1},~\ref{sim2},~\ref{sim3},~\ref{sim4} and~\ref{sim6}.
\end{proof}

%

\subsection{Quantum bounded storage adversary}\label{qqbsadv}
Kasher and Kempe~\cite{KK10} introduced the quantum bounded storage adversary ($\qbsa$) model, where the adversary obtains quantum side information of bounded memory from both sources. 
\begin{definition}[$(b_1, b_2) $-$\qbsa$~\cite{KK10}]\label{qbsadv} 
	Let $\tau_{X\hat{X}}$, $\tau_{Y\hat{Y}}$ be the canonical purifications of independent sources $X,Y$ respectively (registers $\hat{X}\hat{Y}$ with Reference). 
	\begin{enumerate}
		\item  Alice and Bob hold $X,Y$ respectively. They also share an entangled pure  state $\phi_{NM}$ (Alice holds $N$, Bob holds $M$).
		
		\item Alice applies a CPTP map $\psi_A : \mathcal{L} (\cH_{X} \otimes \cH_{N}) \rightarrow   \mathcal{L}(\cH_{X} \otimes \cH_{N'})$ and Bob applies a CPTP map $\psi_B :    \mathcal{L} (\cH_Y \otimes \cH_{M}) \rightarrow   \mathcal{L}(\cH_{Y} \otimes \cH_{M'})$. Let $$\rho_{X\hat{X}N'M'Y\hat{Y}} = (\psi_A \otimes \psi_B) (\tau_{X\hat{X}} \otimes \phi_{NM} \otimes \tau_{Y\hat{Y}}).$$

		\item  Adversary gets access to $\rho_{N'M'}$ with $ \log \dim(\cH_{N'}) \leq b_1$, $ \log \dim(\cH_{M'}) \leq b_2$ respectively.
	\end{enumerate}
\end{definition}
We show how to simulate a $(b_1, b_2)$-$\qbsa$ in the model of an $l\mhyphen\qma$.
\begin{claim}
	\label{sim1}
	A $(b_1, b_2) $-$\qbsa$ acting on an $(n,k_1,k_2)$-source can be simulated by an $l\mhyphen\qma$ for some $l \leq 2\min \{b_1,b_2 \} +2n-k_1-k_2$.
\end{claim}
\begin{proof}Let $V_A: \mathcal{H}_X \otimes \mathcal{H}_{N} \rightarrow \mathcal{H}_{X} \otimes \mathcal{H}_{N'} \otimes \mathcal{H}_{\hat{N'}}$, $V_B: \mathcal{H}_Y \otimes \mathcal{H}_{M} \rightarrow \mathcal{H}_{Y} \otimes \mathcal{H}_{M'} \otimes \mathcal{H}_{\hat{M'}}$ be the Stinespring isometry extensions of CPTP maps $\psi_A$, $\psi_B$ respectively i.e. $\psi_A(\theta)=  \tr_{\hat{N'}  }(V_A \theta V_A^\dagger)$ for every state $\theta_{XN}$ and  $\psi_B(\theta)=  \tr_{\hat{M'}  }(V_B \theta V_B^\dagger)$ for every state $\theta_{YM}$. Let
	$$\rho^A_{X\hat{X}N'\hat{N'} M} = (V_A \otimes \id) (\tau_{X\hat{X}}  \otimes \phi_{NM})(V_A \otimes \id)^\dagger,$$  $$\rho_{X\hat{X}N'\hat{N'} M'\hat{M'} Y\hat{Y}} = (V_A \otimes V_B) (\tau_{X\hat{X}} \otimes \phi_{NM} \otimes \tau_{Y\hat{Y}})(V_A \otimes V_B)^\dagger.$$Note $\rho^A_{\hat{X}N' M}$ is such that $\rho^A_{\hat{X} M} = \tau_{\hat{X}} \otimes \phi_M.$ Thus, from Fact~\ref{fact:boundnew}, we have $\imax(\hat{X}:N' M)_{\rho^A} \leq 2\log ( \dim(\cH_{N'}))  \leq  2b_1$. Let $\sigma_{N'M}$ be such that 
	\[ \dmax{\rho^A_{\hat{X}N' M}}{ \rho^A_{\hat{X} } \otimes \sigma_{N'M}}  \leq  2b_1 .\]Also, note $\rho^A_{\hat{X}} = \tau_{\hat{X}}  \leq 2^{-k_1} \id_{\hat{X}}.$ The inequality follows since the min-entropy of $\tau_X$ is at least $k_1$ and $\tau_{\hat{X}}$ is canonical purification of $\tau_X$. Thus, we further have 
	\[ \dmax{\rho^A_{\hat{X}MN'}}{ \id_{\hat{X}} \otimes \sigma_{MN'}  } \leq  2b_1-k_1 .\]
	Thus $\hmin{X}{N'M}_{\rho^A} =\hmin{\hat{X}}{N'M}_{\rho^A}\geq k_1-2b_1$. Note the first equality is because $\rho^A_{XN'M} =\rho^A_{\hat{X}N'M}$. Using Fact~\ref{fact102}, we further have 
	$$\hmin{X}{N'M'\hat{M'} Y\hat{Y}}_{\rho} \geq \hmin{X}{N'M}_{\rho^A} \geq k_1-2b_1.$$Note $ \dmax{\rho_{\hat{Y}X\hat{X} \hat{N'}}}{ U_{\hat{Y}} \otimes \rho_{X\hat{X}\hat{N'}}} \leq  n-k_2 $ since $\rho_{\hat{Y}X\hat{X}\hat{N'}} = \tau_{\hat{Y} } \otimes  \rho_{X\hat{X}\hat{N'}}$, the min-entropy of $\tau_Y$ is at least $k_2$ and $\tau_{\hat{Y}}$ is canonical purification of $\tau_Y$. Thus $\hminn{Y}{X\hat{X}\hat{N'} }_\rho =\hminn{\hat{Y}}{X\hat{X}\hat{N'} }_\rho \geq k_2$. Note the first equality is because $\rho_{YX\hat{X}\hat{N'}} =\rho_{\hat{Y}X\hat{X}\hat{N'}}$.
	
	Simulation of the state $\rho$ in the model of $l\mhyphen\qma$ follows from Lemma~\ref{lemma:simeverything}. Using Lemma~\ref{lemma:simeverything},  with the following assignment of registers (below the registers on the left are from Lemma~\ref{lemma:simeverything} and the registers on the right are the registers in this proof)
 $$(X, Y, \hat{X}, \hat{Y},  N, M, \rho) \leftarrow (X, Y, \hat{X}, \hat{Y}, \hat{N'}, N'M'\hat{M'}, \rho),$$we have $l \leq   2n+2b_1-k_1-k_2.$

	A similar argument can be given by exchanging the roles of Alice and Bob. The desired follows.
\end{proof}

As a corollary, we reproduce the security of one-bit output  inner-product against $(b_1,b_2)\mhyphen \qbsa$ acting on $(n,k_1,k_2)$-source from~\cite{KK10} as follows.
\begin{corollary}[\cite{KK10}]
An $(n,n,1)$-$2$-source extractor $\IP^n_2 : \{0,1\}^n \times \{0,1\}^n \to \{0,1\}$ is   $(b_1,b_2,\eps)$-quantum secure against $\qbsa$ on $(n,k_1,k_2)$-source for parameters $k_1+k_2- 2 \min{ \{b_1,b_2\}} \geq n+41 +8 \log(1/\eps).$

\end{corollary}
\begin{proof}
The proof follows from Claim~\ref{sim1} (noting $l \mhyphen \qma$ is also a $(k_1-2b_1,k_2) \mhyphen \qpa$ or $(k_1,k_2-2b_2) \mhyphen \qpa$) and Corollary~\ref{corr:iphminhmin}.
\end{proof}

%

\subsection{Quantum independent adversary}\label{qqiadv}
Kasher and Kempe~\cite{KK10} introduced the  quantum independent adversary ($\qia$) model, where the adversary obtains independent side information from both sources. 
\begin{definition}[$(k_1, k_2)$-$\qia$,~$(k_1,k_2)\mhyphen\qias$~\cite{KK10}]\label{qiadv} 
	Let $\tau_{X\hat{X}}$, $\tau_{Y\hat{Y}}$ be the canonical purifications of independent sources $X,Y$ respectively (registers $\hat{X}\hat{Y}$ with Reference). 
	\begin{enumerate}
		\item  Alice and Bob hold $X,Y$ respectively. They also share a product state $\phi_{NM}= \phi_{N} \otimes \phi_{M}$ (Alice holds $N$ and Bob holds $M$).
		\item Alice applies CPTP map $\psi_A : \mathcal{L} (\cH_{X} \otimes \cH_{N}) \rightarrow   \mathcal{L}(\cH_{X} \otimes \cH_{N'})$ and Bob applies CPTP map $\psi_B :    \mathcal{L} (\cH_Y \otimes \cH_{M}) \rightarrow   \mathcal{L}(\cH_{Y} \otimes \cH_{M'})$. Let \[ \rho_{X\hat{X}N'M'Y\hat{Y}} = (\psi_A \otimes \psi_B) (\tau_{X\hat{X}} \otimes \phi_{NM} \otimes \tau_{Y\hat{Y}}) \quad \quad  ;  \quad \quad  \rho_{X\hat{X}N'M'Y\hat{Y}}  = \rho_{X\hat{X}N'}  \otimes \rho_{M'\hat{Y}Y}, \]
		with $\hmin{X}{N'}_\rho \geq k_1$ and $\hmin{Y}{M'}_\rho \geq k_2$.
		\item  Adversary gets access to $\rho_{N'M'}$. The state $\rho$ is called a $(k_1,k_2)\mhyphen\qias$. 
	\end{enumerate}
\end{definition}
We show how to simulate  a $(k_1, k_2) $-$\qia$ in the model of an $l\mhyphen\qma$.
\begin{claim}
	\label{sim2}
	A $(k_1, k_2) $-$\qias$ is an $l\mhyphen\qmas$ for some $l \leq 2n-k_1-k_2$.
\end{claim}
\begin{proof}

	Let $V_A: \mathcal{H}_X \otimes \mathcal{H}_{N} \rightarrow \mathcal{H}_{X} \otimes \mathcal{H}_{N'} \otimes \mathcal{H}_{\hat{N'}}$, $V_B: \mathcal{H}_Y \otimes \mathcal{H}_{M} \rightarrow \mathcal{H}_{Y} \otimes \mathcal{H}_{M'} \otimes \mathcal{H}_{\hat{M'}}$ be the Stinespring isometry extensions of CPTP maps $\psi_A$, $\psi_B$ respectively. Let $\phi_{N \hat{N} M \hat{M}} =\phi_{N \hat{N} } \otimes \phi_{M \hat{M} }$ be the purification of $\phi_{NM}$. Let $$\rho_{X\hat{X}N'\hat{N'} \hat{N} M'\hat{M'}\hat{M} Y\hat{Y}} = (V_A \otimes V_B) (\tau_{X\hat{X}} \otimes \phi_{N \hat{N} M \hat{M}} \otimes \tau_{Y\hat{Y}})(V_A \otimes V_B)^\dagger.$$Since $ \rho_{X\hat{X}N'\hat{N'} \hat{N}M'\hat{M'} \hat{M}Y\hat{Y}}  = \rho_{X\hat{X}N'\hat{N'} \hat{N}}  \otimes \rho_{M'\hat{M'} \hat{M}\hat{Y}Y},$ we have the conditional-min-entropy bound  $$\hmin{X}{Y\hat{Y}M'\hat{M'}\hat{M}N'}_\rho =\hmin{X}{N'}_\rho \geq k_1.$$
	Also, since $\hmin{Y}{M'}_\rho \geq k_2$, using Fact~\ref{fact102}, we have $\Hmin({Y})_\rho \geq k_2$. Since $\rho_{YX\hat{X}\hat{N'}\hat{N}}  = \rho_Y \otimes \rho_{X\hat{X}\hat{N'}\hat{N}},$ we have 
	\[  \dmax{\rho_{YX\hat{X}\hat{N'}\hat{N}}  }{U_Y \otimes \rho_{X\hat{X}\hat{N'}\hat{N}} }  \leq  \log ( \dim(\cH_{Y}) ) -k_2=n-k_2.\]
	Thus, $\hminn{Y}{X\hat{X}\hat{N'}\hat{N}}_\rho \geq k_2$.

Simulation of the state $\rho$ in the model of $l\mhyphen\qma$ follows from Lemma~\ref{lemma:simeverything}. Using Lemma~\ref{lemma:simeverything},  with the following assignment of registers (below the registers on the left are from Lemma~\ref{lemma:simeverything} and the registers on the right are the registers in this proof)
 $$(X, Y, \hat{X}, \hat{Y},  N, M, \rho) \leftarrow (X, Y, \hat{X}, \hat{Y}, \hat{N'}\hat{N}, M'\hat{M'}\hat{M}N', \rho),$$we have $l \leq   2n-k_1-k_2.$

\end{proof}

\subsection{General entangled adversary}\label{qgeadv}
 Chung, Li and Wu~\cite{CLW14} introduced general entangled adversary ($\gea$) defined as follows:
\begin{definition}[$(k_1, k_2)$-$\gea$, $(k_1,k_2)\mhyphen\geas$~\cite{CLW14}]\label{geadv} 
	Let $\tau_{X\hat{X}}$, $\tau_{Y\hat{Y}}$ be the canonical purifications of independent sources $X,Y$ respectively (registers $\hat{X}\hat{Y}$ with Reference).  
	\begin{enumerate}
		\item  Alice and Bob hold $X,Y$ respectively. They also hold entangled pure state $\phi_{NM}$ (Alice holds $N$, Bob holds $M$).
		\item  Alice applies a CPTP map $\psi_A :  \mathcal{L}(  \cH_{X} \otimes \cH_{N}) \rightarrow   \mathcal{L}(  \cH_{X} \otimes \cH_{N'})$ and Bob applies a CPTP map $\psi_B :     \mathcal{L}(  \cH_Y \otimes \cH_{M}) \rightarrow   \mathcal{L}(  \cH_{Y} \otimes \cH_{M'})$. Let \[ \rho^A_{X\hat{X}N'MY\hat{Y}} = (\psi_A \otimes \id_{}) (\tau_{X\hat{X}} \otimes \phi_{NM} \otimes \tau_{Y\hat{Y}})  ,\] 
		\[   \rho^B_{X\hat{X}NM'Y\hat{Y}} = (\id_{} \otimes \psi_B) (\tau_{X\hat{X}} \otimes \phi_{NM} \otimes \tau_{Y\hat{Y}}) ,\]  
		$$\rho_{X\hat{X}N'M'Y\hat{Y}} = (\psi_A \otimes \psi_B) (\tau_{X\hat{X}} \otimes \phi_{NM} \otimes \tau_{Y\hat{Y}})  =  (\id_{} \otimes \psi_B) \rho^A_{XN'MY},$$
		with $\hmin{X}{N'M}_{\rho^A} \geq k_1$ and $\hmin{Y}{NM'}_{\rho^B} \geq k_2$. 
		\item  Adversary gets access to $\rho_{N'M'}$. The state $\rho$ is called a $(k_1,k_2)\mhyphen\geas$. 
	\end{enumerate}
\end{definition}
We show how to simulate a $(k_1, k_2)$-$\gea$ in the model of an $l\mhyphen\qma$.
\begin{claim}
	\label{sim3}
	A $(k_1, k_2) $-$\geas$ is an $l\mhyphen\qmas$ for some $l \leq 2n-k_1-k_2$.
\end{claim}
\begin{proof}
	
	Let $V_A: \mathcal{H}_X \otimes \mathcal{H}_{N} \rightarrow \mathcal{H}_{X} \otimes \mathcal{H}_{N'} \otimes \mathcal{H}_{\hat{N'}}$, $V_B: \mathcal{H}_Y \otimes \mathcal{H}_{M} \rightarrow \mathcal{H}_{Y} \otimes \mathcal{H}_{M'} \otimes \mathcal{H}_{\hat{M'}}$ be the Stinespring isometry extensions of CPTP maps $\psi_A$, $\psi_B$ respectively.  Let $$\rho_{X\hat{X}N'\hat{N'}  M'\hat{M'} Y\hat{Y}} = (V_A \otimes V_B) (\tau_{X\hat{X}} \otimes \phi_{N M } \otimes \tau_{Y\hat{Y}})(V_A \otimes V_B)^\dagger.$$Let
	\[\rho^A_{X\hat{X}N'\hat{N'} M} = (V_A \otimes \id) (\tau_{X\hat{X}}  \otimes \phi_{NM})(V_A \otimes \id)^\dagger \quad ; \quad \rho^B_{NY\hat{Y}M'\hat{M'}} = (\id \otimes V_B) (\phi_{NM} \otimes \tau_{Y\hat{Y}}  )(\id \otimes V_B)^\dagger .\]
	Note $\rho^A_{X\hat{X}N'\hat{N'} M}$ is the purification of $\rho^A_{X\hat{X}N'M}$. From $\hmin{X}{N'M}_{\rho^A} \geq k_1$ and
	from Fact~\ref{fact102} it follows that  $$\hmin{X}{N'M'\hat{M'}Y\hat{Y} }_{\rho} \geq k_1.$$

		Also, since $\hmin{Y}{NM'}_{\rho^B} \geq k_2$, using Fact~\ref{fact102}, we have $\Hmin({Y})_{\rho^B} \geq k_2$. Noting $\rho^B_{YN}  = \rho^B_Y \otimes \rho^B_{N}$ (since $V_B$ is safe on register $Y$), we have 
	\[  \dmax{\rho^B_{YN}  }{U_Y \otimes \rho^B_{N} }  \leq  \log ( \dim(\cH_{Y}) ) -k_2=n-k_2.\]
	Using Fact~\ref{data}, we have \[  \dmax{\rho_{YX\hat{X} \hat{N'}}  }{U_Y \otimes \rho_{X\hat{X}\hat{N'}} }  \leq n-k_2.\]
	Thus, $\hminn{Y}{X\hat{X}\hat{N'}}_\rho \geq k_2$.

Simulation of the state $\rho$ in the model of $l\mhyphen\qma$ follows from Lemma~\ref{lemma:simeverything}. Using Lemma~\ref{lemma:simeverything},  with the following assignment of registers (below the registers on the left are from Lemma~\ref{lemma:simeverything} and the registers on the right are the registers in this proof)
 $$(X, Y, \hat{X}, \hat{Y},  N, M, \rho) \leftarrow (X, Y, \hat{X}, \hat{Y}, \hat{N'}, N'M'\hat{M'}, \rho),$$we have $l \leq   2n-k_1-k_2.$

\end{proof}
\subsection{Quantum Markov adversary}\label{sec:madv}
Arnon-Friedman, Portmann and Scholz~\cite{APS16} introduced quantum Markov adversary\newline($\qmra$), such that  adversary's side information forms a Markov-chain with the both sources.
\begin{definition}[$(k_1, k_2)$-$\qmra$, $(k_1,k_2)\mhyphen \qmras$)~\cite{APS16}]\label{madv} 
	Let $\rho_{XEY}$ be a Markov-chain $(X-E-Y)_\rho$ with $\hmin{X}{E}_{\rho}  \geq k_1$ and $\hmin{Y}{E}_{\rho} \geq k_2$. Adversary gets access to quantum register $E$. The state $\rho$ is called a $(k_1,k_2)\mhyphen \qmras$.
\end{definition}
\begin{claim}
	\label{sim4}
	A $(k_1, k_2) $-$\qmras$ can be simulated $\eps$-approximately by an $l\mhyphen\qma$ for some $l \leq 2n-k_1-k_2+25+6 \log (1/\eps)$.
\end{claim}
\begin{proof}
	From Fact~\ref{fact:markov}, 
	$$\rho_{XEY} = \sum_{t}  \Pr(T=t) \ketbra{t} \otimes  \left(\rho^t_{XE_1} \otimes \rho^t_{YE_2} \right),$$ where $T$ is some classical register over finite alphabet. Let $\rho_{X\hat{X}T\hat{T} E_1 \hat{E_1}E_2 \hat{E_2}Y\hat{Y}}$ be a pure state extension of $\rho_{XEY} \equiv \rho_{XE_1TE_2Y}$ such that,
	\[\rho_{X\hat{X}T\hat{T} E_1 \hat{E_1}E_2 \hat{E_2}Y\hat{Y}} = \sum_{t}\sqrt{\Pr(T=t)} \ket{tt}_{T\hat{T}} \ket{\rho}_{X\hat{X}E_1 \hat{E_1}E_2\hat{E_2}Y\hat{Y}}^t,\]
	\[\hmin{X}{E}_\rho \geq k_1 \quad ; \quad \hmin{Y}{E}_\rho \geq k_2, \]
	registers ($XYT$) are classical (with copies $\hat{X}\hat{Y}\hat{T}$) and  $\ket{\rho}_{X\hat{X}E_1 \hat{E_1}E_2\hat{E_2}Y\hat{Y}}^t = \ket{\rho}_{X\hat{X}E_1 \hat{E_1}}^t \otimes \ket{\rho}_{E_2\hat{E_2}Y\hat{Y}}^t$. Additionally, note for every $T=t$, $\ket{\rho}_{X\hat{X}E_1 \hat{E_1}}^t \otimes \ket{\rho}_{E_2\hat{E_2}Y\hat{Y}}^t$ is the pure state extension of $\rho^t_{XE_1} \otimes \rho^t_{YE_2}$ with $\ket{\rho}_{X\hat{X}E_1 \hat{E_1}}^t, \ket{\rho}_{E_2\hat{E_2}Y\hat{Y}}^t$ canonical purifications of  $\rho^t_{XE_1}, \rho^t_{YE_2}$ respectively.  Since $E \equiv E_1TE_2$, using Fact~\ref{data}, we have 
	\begin{equation}\label{eq:corr3111}
		\hmin{X}{E_1T}_\rho \geq \hmin{X}{E}_\rho \geq k_1 \quad ; \quad \hmin{Y}{E_2T}_\rho \geq \hmin{Y}{E}_\rho \geq k_2.
	\end{equation}Consider,
	\begin{align*}
		\hmin{X}{Y\hat{Y}{E}_2E_1{T}}_\rho & \geq \hmin{X}{Y\hat{Y}\hat{E}_2E_2E_1{T}}_\rho  \\
		& = \hmin{X}{E_1T}_\rho  \\
		& \geq k_1. 
	\end{align*}
	The first equality is because conditioned on every $T=t$, $\rho^t_{XE_1E_2\hat{E}_2Y\hat{Y}} =\rho^t_{XE_1} \otimes \rho^t_{E_2\hat{E}_2Y\hat{Y}}$. The first inequality follows from Fact~\ref{fact102} and second inequality follows from Eq.~\eqref{eq:corr3111}. Consider,
	\begin{align*}
		\hmin{Y}{X\hat{X} \hat{E}_2 \hat{E}_1 \hat{T}}_\rho & \geq 	\hmin{Y}{X\hat{X} \hat{E}_2 E_1\hat{E}_1 \hat{T}}_\rho  \\
		& = \hmin{Y}{ \hat{E}_2 \hat{T}}_\rho  \\
		& = \hmin{Y}{{E}_2T}_\rho  \\
		& \geq k_2. 
	\end{align*}
	The first equality is because conditioned on every $\hat{T}=t$, $\rho^t_{Y E_1\hat{E}_1\hat{E}_2X\hat{X}} =\rho^t_{Y\hat{E}_2} \otimes \rho^t_{E_1\hat{E}_1X\hat{X}}$. The second equality is because $\rho_{Y\hat{E}_2\hat{T}} \equiv \rho_{YE_2T}$. The first inequality follows from Fact~\ref{fact102} and second inequality follows from Eq.~\eqref{eq:corr3111}.

	For the state $\rho$ with the following assignment (terms on the left are from Definition~\ref{qmadvk1k2} and on the right are from here),
	\[(X,\hat{X},N,M,Y,\hat{Y}) \leftarrow (X,\hat{X}, \hat{E}_2 \hat{E}_1\hat{T},{E}_2E_1{T},Y,\hat{Y}),\] we have $\rho$ is a $(k_1,k_2) \mhyphen \qpas$. Using Fact~\ref{lemma:nearby_rho_prime_prime4}, we have an $l \mhyphen \qmas$ $\rho'$ such that 
	$l \leq 2n-k_1-k_2+25+6 \log (1/\eps)$ and $\rho' \approx_{\eps} \rho.$  
\end{proof}
\subsection{Quantum communication adversary}\label{qcadv}
We show how to simulate  a $(k_1, k_2)$-$\qca$ (see Definition \ref{def:qcadv}) in the model of an  $l\mhyphen\qma$.
\begin{claim}
	\label{sim6}
	A $(k_1, k_2) $-$\qcas$ can be simulated $\eps$-approximately by an $l\mhyphen\qma$ for some $l \leq 2n-k_1-k_2+25+6 \log (1/\eps)$.
\end{claim}
\begin{proof}
Let $\Phi_{X\hat{X}N'M'Y\hat{Y}}$, be the end state after the action of $(k_1, k_2)$-$\qca$ (adversary gets registers either $M'Y$ or $N'X$ of his choice).

 We note to the reader that every $(k_1,k_2) \mhyphen \qcas$ is  a  
$(k_1,k_2)\mhyphen\qpas$.  Now using Fact~\ref{lemma:nearby_rho_prime_prime4}, we have an $l \mhyphen \qmas$ $\rho'$ such that 
	$l \leq 2n-k_1-k_2+25+6 \log (1/\eps)$ and $\rho' \approx_{\eps} \rho.$  This completes the proof. 
\end{proof}

\section{A quantum secure weak-seeded non-malleable extractor and privacy-amplification}\label{sec:app}

\subsection{A quantum secure weak-seeded  non-malleable extractor}\label{qnmadv}
\begin{theorem}\label{nmext_new}
		Let $p \ne 2$ be a prime, $n$ be an even integer and $\eps > 0$. The function  $\nmext(X,Y)$ as defined in Definition~\ref{ipnme} is a $(k_1,k_2)$-quantum secure weak-seeded non-malleable extractor against $\nma$ for the parameters $k_1 + k_2 \geq (n+17) \log p +33+16 \log \left( \frac{1}{\eps}  \right)$.
\end{theorem}
\begin{proof}Let  $\rho$ be a $(k_1,k_2)\mhyphen\nmasw$. 
From \cite{ACLV18} (see Theorem~1), to show,
$$  \| \rho_{ \nmext(X,Y)\nmext(X,Y') YY'M'} - U_{\log p}\otimes \rho_{ \nmext(X,Y') YY'M'} \|_1 \leq \eps, $$
for $\nmext(X,Y)$ as defined in Definition~\ref{ipnme} it is enough to show, 
$$  \| \rho_{ \langle X,g(Y,Y') \rangle YY'M'} - U_{\log p} \otimes \rho_{YY'M'} \|_1 \leq  \frac{2\eps^2}{p^2},$$ where $g : \mathbb{F}^{n/2}_p \times  \mathbb{F}^{n/2}_p \to  \mathbb{F}^n_p$ is an (appropriately defined) function such that
for any $z \in  \mathbb{F}^n_p$ there are at most two possible pairs $(y,y')$ and $y \ne y'$ such that $g(y,y')=z$. 
Let  $U: \cH_Y \otimes \cH_{Y'} \rightarrow \cH_Y \otimes \cH_{Y'} \otimes \cH_Z \otimes \cH_{\hat{Z}}$ be a safe isometry such that $Z=g(Y,Y')$.  Let  $\theta=U \rho U^\dagger.$  Noting $\theta_{XM'YY'} =\rho_{XM'YY'}$, we have $ \| \theta_{ \langle X,Z \rangle YY'M'} - U_{\log p} \otimes \theta_{YY'M'} \|_1 =  \| \rho_{ \langle X,g(Y,Y') \rangle YY'M'} - U_{\log p} \otimes \rho_{YY'M'} \|_1.$ From Claim~\ref{comm_game}, for parameters
$$ k_1+k_2 \geq (n+1) \log p + 41+8  \log \left(\frac{p^2}{2\eps^2} \right),$$we get 
$\Vert \theta_{\langle X,Z \rangle YM'Y'} - U_{\log p} \otimes  \theta_{YM'Y'} \Vert_1 \leq \frac{2\eps^2}{p^2}.$ Rearranging terms we get the desired which completes the proof.

\end{proof}

\begin{claim}
\label{comm_game}
 $\Vert \theta_{\langle X,Z \rangle YM'Y'} - U_{\log p} \otimes  \theta_{YM'Y'} \Vert_1 \leq \eps'',$  for parameters, $$k_1+k_2 \geq (n+1) \log p + 41+8  \log \left(\frac{1}{\eps''} \right).$$
\end{claim}

\begin{proof}

Note $\theta$ is a $(k_1,k_2)$-$\nmasw$. Note,  
 $$\hmin{X}{M'YY'\hat{Y}\hat{Y}'Z\hat{Z}}_{\theta} \geq k_1.$$
Also,
$$\dmax{ \rho_{YY'NX\hat{X}}}{\id_{YY'}\otimes \rho_{NX\hat{X}}}  \leq   \dmax{ \rho_{YNX\hat{X}}}{\id_Y  \otimes \rho_{NX\hat{X}}}  \leq -k_2.$$
The first inequality follows from Fact~\ref{rhoablessthanrhoaidentity}.
 The second inequality follows since $\Hmin(Y)_{\rho} \geq k_2$ and $\rho_{YNX\hat{X}} =\rho_{Y} \otimes \rho_{NX\hat{X}}$. Thus,
$$ \rho_{YY'NX\hat{X}} \leq 2^{-k_2+ n \log p } (U_{YY'}\otimes \rho_{NX\hat{X}}),$$
which from Fact~\ref{data} implies 
$$ \theta_{ZNX\hat{X}} \leq 2^{-k_2+ n \log p }  \left(\tau_{Z}\otimes \theta_{NX\hat{X}} \right),$$
where $\tau_Z = g(U_{YY'})$.
Since there are at most two possible pairs $(y,y')$ and $y \ne y'$ such that $g(y,y')=z$, we have $\tau_{Z} \le 2 \cdot U_Z$.
Thus,
$$   \dmax{ \theta_{ZNX\hat{X}}}{U_Z \otimes \theta_{NX\hat{X}}}  \leq  n \log p -k_2+1.$$ This implies, $$\hmin{Z}{NX\hat{X}}_{\theta} \geq \hminn{Z}{NX\hat{X}}_{\theta} \geq k_2-1.$$

We note that the state $\theta$ is a $(k_1,k_2-1)\mhyphen \qpas$.  Using Corollary~\ref{corr:iphminhmin},  with the following assignment of registers (below the registers on the left are from Corollary~\ref{corr:iphminhmin} and the registers on the right are the registers in this proof)
 $$(X, Y, \hat{X}, \hat{Y},  N, M, \rho) \leftarrow (X, Z, \hat{X}, \hat{Z}, N, M'YY'\hat{Y}\hat{Y}', \theta),$$ we get 
$\Vert \theta_{\langle X,Z \rangle YM'Y' \hat{Y}\hat{Y}'Z} - U_{\log p} \otimes  \theta_{YM'Y' \hat{Y}\hat{Y}'Z} \Vert_1 \leq \eps''$ for parameters, $$k_1+k_2 \geq (n+1) \log p + 41+8  \log \left(\frac{1}{\eps''} \right).$$

Since $\Vert \theta_{\langle X,Z \rangle YM'Y'} - U_{\log p} \otimes  \theta_{YM'Y'} \Vert_1 \leq \Vert \theta_{\langle X,Z \rangle YM'Y' \hat{Y}\hat{Y}'Z} - U_{\log p} \otimes  \theta_{YM'Y' \hat{Y}\hat{Y}'Z} \Vert_1,$ we get the desired.
\end{proof}
\subsection{Privacy-amplification with local weak-sources}


 Let $n, m, d, l$ be positive integers and $k, k_1, k_2, \eps, \delta>0$. We start with the definition of a quantum secure privacy-amplification protocol against active adversaries. The following description is from~\cite{ACLV18}. 
A privacy-amplification protocol $(P_A, P_B)$ is defined as follows. The protocol is executed by two parties Alice and Bob sharing a
secret $X\in \{0,1\}^n$, whose actions are described by $P_A$, $P_B$ respectively~\footnote{It is not necessary for the definition to specify exactly how the protocols are formulated; informally, each player's actions is described by a sequence of efficient algorithms that compute the player's next message, given the past interaction.}. In addition there is an active, computationally
unbounded adversary Eve, who might have some quantum side information $E$
correlated with $X$ but satisfying $\Hmin(X \vert E)_{\rho} \ge k$, where $\rho_{XE}$ denotes the initial state at beginning of the protocol. 

Informally, the goal for the protocol is that
whenever a party (Alice or Bob) does not reject, the key $R$ output by
this party is random and statistically independent of Eve's
view. Moreover, if both parties do not reject, they must output the
same keys $R_A=R_B$ with overwhelming probability.

More formally, we assume that Eve is in full control of the
communication channel between Alice and Bob, and can arbitrarily
insert, delete, reorder or modify messages sent by Alice and Bob to
each other. 
At the end of the
protocol, Alice outputs a key $R_A\in
\{0,1\}^l \cup \{\perp\}$, where $\perp$ is a special symbol indicating
rejection. Similarly, Bob outputs a key $R_B \in \{0,1\}^l \cup
\{\perp\}$. For a random variable $R\in
\{0,1\}^l \cup \{\perp\}$, let $\mathsf{purify}(R)$ be a random variable on $l$-bit strings that is deterministically equal to $\perp$ if $R=\perp$, and is otherwise uniformly distributed over $\lbrace 0, 1\rbrace ^l$. The following definition generalizes the classical definition in \cite{DLWZ14}.

\begin{definition}\label{privamp}
	Let $k,l$ be integer and $\eps > 0$. A privacy-amplification protocol $(P_A, P_B)$ is
	a $(k, l, \eps)$-\emph{privacy-amplification protocol secure against active quantum adversaries} if it satisfies
	the following properties for any initial state $\rho_{XE}$ such that $\Hmin(X|E)_\rho \geq k$, and $\sigma$ being the joint state of Alice, Bob and Eve at the end of the protocol given by $(P_A,P_B)$ including $\mathsf{purify}(R_A)$ and $\mathsf{purify}(R_B)$. 
	\begin{enumerate}
		\item \emph{Correctness.} If the adversary does not interfere with the protocol, then $\Pr[R_A=R_B \land~ R_A\neq \perp \land~ R_B\neq \perp]_\sigma=1$. 
		\item \emph{Robustness.} This property states that even in the presence of an active adversary, $\Pr[R_A \neq R_B \land~ R_A \neq \perp \land~ R_B \neq \perp]_\sigma \le \eps$. 
		
		\item \emph{Extraction.}  Let $\sigma_{\tilde{E}}$ be the final quantum state possessed by Eve (including the transcript of the protocol). The following should hold: 
		\[\Vert \sigma_{R_A  \tilde{E}}- \sigma_{\mathsf{purify}(R_A) \tilde{E}} \Vert_1 \leq \eps		~~~~\mbox{and}~~~~
	\Vert \sigma_{R_B \tilde{E}}- \sigma_{\mathsf{purify}(R_B) \tilde{E}} \Vert_1 \leq \eps\;.\]
		
			In other words, whenever a party does not reject, the party's key is (approximately) indistinguishable from a fresh random string to the adversary.
	
		%
		%
	\end{enumerate}
	The quantity $k-l$ is called the \emph{entropy loss}. 
\end{definition}
\paragraph{Our protocol.} 
\begin{Protocol}[htb]
	\begin{center}
		\begin{tabular}{l c l}
			Alice:  $X$ & Eve: $E$ & ~~~~~~~~~~~~Bob: $X$ \\
			
			\hline\\
			Sample local weak-source $A$& &Sample local weak-source $B$\\
			& $A \longrightarrow A'$ & \\
			&& \\
			$S = \mathsf{nmExt}(X,A)$ &&  $S' = \mathsf{nmExt}(X,A')$\\
			&&  $W' = \mathsf{Ext}( X, B)$\\
			& $W,T \longleftarrow W',T'=\mathsf{MAC}_{S'}(W')$ & \\

			{\bf If} $T \neq \mathsf{MAC}_{S}(W)$ {\em reject}&&\\
			Set final $R_A = \mathsf{Tre}(X,W)$&& Set final $R_B = \mathsf{Tre}(X,W')$\\
			\hline
		\end{tabular}
		{\small {\caption{\label{prot:priv-amp-DW}
					New $2$-round privacy-amplification protocol with weak local sources.}}
		}
	\end{center}
\end{Protocol}
In Protocol~\ref{prot:priv-amp-DW}, we describe a slight modification of the DW protocol~\cite{DW09}, that achieves PA in the setting where Alice and Bob only have weak local randomness $A$ and $B$ respectively, such that $A, B, (X,E)$ are independent. Let $A$ be $(n/2,k_1) \mhyphen$source and $B$ be $(n,k_2) \mhyphen$source.  We need the following primitives.


	
\begin{itemize}
	\item Let $\nmext:\{0,1\}^n \times\{0,1\}^{n/2} \to \{0,1\}^{2m}$ be  a $(k,k_1,\eps)$-quantum secure  $(n,n/2,2m)$-weak-seeded non-malleable extractor against $\nma$. Then we have from Theorem~\ref{nmext_new}, $$k +k_1 \geq n+34m+33+16 \log \left( \frac{1}{\eps}  \right).$$ 
	\item Let $\Ext:\{0,1\}^n\times\{0,1\}^{n} \to \{0,1\}^m$ be $\IP_{2^m}^{n/m}$. For any state $\sigma_{XEB}$ such that $\sigma_{XEB} =\sigma_{XE} \otimes \sigma_B$, $\hmin{X}{E}_\sigma \geq k$ and $\Hmin(B)_\sigma \geq k_2$, we have 
	$$   \| \sigma_{\Ext(X,B)XE}- U_m \otimes \sigma_{XE}   \|_1 \leq \eps , $$
	for the parameters $k +k_2 \geq n+m+40+8 \log \left( \frac{1}{\eps}  \right)$ from Corollary~\ref{corr:wqma}. 
	\item Let $\mathsf{MAC}:\{0,1\}^{2m} \times \{0,1\}^{m} \to \{0,1\}^m$ be a one-time $ 2^{-m}$-information-theoretically secure message authentication code from Fact~\ref{prop:mac} for $m= O \left(\log^3(\frac{n}{\eps}) \right)$. Note $2^{-m} \le \eps$. 
	\item Let $\mathsf{Tre}:\{0,1\}^n\times\{0,1\}^{m} \to \{0,1\}^\ell$ be a $(2l,\eps)$-quantum secure strong $(n,m,l)$-seeded extractor for $m = O(\log^3(n/\eps))$ from Fact~\ref{trevext}. Taking $l= (1- \delta)k$ for any small constant $\delta>0$ suffices for the privacy-amplification application.
\end{itemize}

In Protocol~\ref{prot:priv-amp-DW}, there are only two messages exchanged, $A$ from Alice to Bob and $(W',T')$ from Bob to Alice. To each of these messages the adversary may apply an arbitrary transformation, that may depend on its side information. 
\begin{definition}[Active attack]
An active attack against PA protocol is described by 3 parameters.
\begin{itemize}
\item A c-q state $\rho_{XE}$ (of adversary choice) such that $\hmin{X}{E}_{\rho} \geq k$.
\item A CPTP map $\mathsf{T}_1: \mathcal{H}_E \otimes \mathcal{H}_{A} \rightarrow \mathcal{H}_{E^\prime} \otimes \mathcal{H}_A \otimes \mathcal{H}_{A^\prime}$.
\item A CPTP map $\mathsf{T}_2: \mathcal{H}_{E^\prime} \otimes \mathcal{H}_{W^\prime} \otimes \mathcal{H}_{T^\prime} \rightarrow \mathcal{H}_{E^{\prime \prime}}  \otimes \mathcal{H}_{W'} \otimes \mathcal{H}_{T'} \otimes \mathcal{H}_{W} \otimes \mathcal{H}_{T}$.
\end{itemize} 
\end{definition}

\begin{theorem}
\label{thm:PA_added1}For any active attack $(\rho_{XE},\mathsf{T}_1,\mathsf{T}_2)$, Protocol~\ref{prot:priv-amp-DW} is $\left(k, \left(1 -\delta \right)k, O({\eps}) \right)$-secure as defined in Definition~\ref{privamp} as long as,
\[ k +k_2 \geq n+m+40+8 \log \left( \frac{1}{\eps}  \right) \quad ;\quad  k +k_1 \geq n+34m+33+16 \log \left( \frac{1}{\eps}  \right) \quad ; \quad m = O(\log^3(n/\eps)).\]

\end{theorem}
\begin{proof}
    The security of Protocol~\ref{prot:priv-amp-DW}  follows from the observation that the protocol is nearly identical to that in~\cite{ACLV18}~\cite{BJK21}. There are two main differences:
\begin{itemize}
	\item The seed $A$ for the non-malleable extractor is a weak-source. This is secure via Theorem~\ref{nmext_new}.
	\item In the protocol from~\cite{ACLV18}~\cite{BJK21}, Bob has immediate access to uniform $W'$. Here, we obtain $W'$ using strong extractor property of inner-product from Corollary~\ref{corr:wqma}. Since $B$ is independent of $X, A$, we get that $W'$ is independent of $X, A$, and hence independent of $X, S$, which is required in the proof of security in~\cite{ACLV18}~\cite{BJK21}. 
\end{itemize}
We refer the reader to Appendix.D~of~\cite{BJK21} for the complete proof of privacy-amplification when the sources $A,B$ are completely uniform. The proof of security follows similar to that of~\cite{BJK21} privacy-amplification protocol, with the following assignment of terms (terms on left are from privacy-amplification protocol of ~\cite{BJK21} and terms on right are from Protocol~\ref{prot:priv-amp-DW}) 
$$(X,Y,S,S',B',T',B,T,R_A,R_B) \leftarrow (X,A,S,S',W',T',W,T,R_A,R_B).$$

\end{proof}

\subsection*{Open questions}\label{sec:conclusion}
\begin{enumerate}
 

	\item  We have shown that the inner-product function is secure against $\qma$. It is natural to ask if the other known  $2$-source extractors are secure against $\qma$. For example is Bourgain's extractor~\cite{Bou05}, which works for $2$-sources with min-entropy $(\frac{1}{2}-\delta)n$ (for constant $\delta$), secure against $\qma$? 
 
	\item Are the multi-source extractor constructions of \cite{Raz05,CZ19,Li19,BIRW06,Li12c,Li11,Li13,R09} secure against (a natural multi-source extension of) $\qma$ ?
	\item We have shown the security of Li's non-malleable extractor against quantum side information. Recently~\cite{BJK21} have shown another non-malleable extractor secure against quantum side information in the $2$-source setting. Several near optimal non-malleable extractor constructions are known in the classical setting, for example constructions of \cite{Li12c,Li17,Li19}. Are they secure against quantum side information? 
	
\end{enumerate} 

	\subsection*{Acknowledgment}\label{sec:acknowledgement}
 The work of NGB was done while he was a graduate student at the Centre for Quantum Technologies, NUS, Singapore. 
 
 This work is supported by the NRF grants NRF-NRFF2013-13 and NRF2021-QEP2-02-P05; the Prime Minister's Office and the Ministry of Education, Singapore, under the Research Centers of Excellence program and grants MOE2012-T3-1-009 and MOE2019-T2-1-145; the NRF2017-NRF-ANR004 {\em VanQuTe} grant and the {\em VAJRA} grant, Department of Science and Technology, Government of India.

\bibliographystyle{alpha}
\bibliography{References}

\end{document}